\documentclass[a4paper,UKenglish,cleveref, autoref, thm-restate]{lipics-v2021}
\nolinenumbers
\hideLIPIcs

\usepackage{amssymb}
\usepackage{amsfonts}
\usepackage{mathtools}
\usepackage{hyperref}
\usepackage{graphicx}
\usepackage{multicol}
\usepackage{amsmath}
\usepackage{stmaryrd}
\usepackage{cleveref}
\usepackage{physics}
\usepackage{booktabs}
\usepackage[ruled,algo2e]{algorithm2e}

\usepackage{tikz}
\usepackage{pgfplots}
\usetikzlibrary{decorations.markings}
\usetikzlibrary{backgrounds,folding}
\usetikzlibrary{shapes.geometric, shapes.misc}
\pgfdeclarelayer{edgelayer}
\pgfdeclarelayer{nodelayer}
\pgfsetlayers{background,edgelayer,nodelayer,main}
\tikzstyle{every picture}=[baseline=-0.25em]
\tikzstyle{none}=[inner sep=0mm]
\tikzstyle{black dot}=[inner sep=1pt,minimum width=0pt,minimum height=0pt,fill=black,draw=black,shape=circle, font=\scriptsize]
\tikzstyle{dot}=[black dot]
\tikzstyle{white dot}=[dot,fill=white]
\tikzstyle{box}=[rectangle,fill=white,draw=black, font=\scriptsize, inner sep=2pt]
\tikzstyle{box-no-outline}=[rectangle, draw=white, fill=white, inner sep=2pt]
\tikzstyle{highlight}=[inner sep=0.8mm,minimum width=0pt,minimum height=0pt,draw=purple,shape=circle]
\tikzset{tickedge/.append style={
			decoration={markings, mark=at position 0.5 with {
					\draw[-] ++ (-2pt,-2pt) -- (2pt,2pt);}
			},postaction={decorate}}
}

\pgfdeclarelayer{edgelayer}
\pgfdeclarelayer{nodelayer}
\pgfsetlayers{background,edgelayer,nodelayer,main}
\tikzstyle{every loop}=[]

\allowdisplaybreaks

\SetKwComment{Comment}{/* }{ */}
\Crefname{algocf}{Algorithm}{Algorithms}

\title{Resource-Efficient Synthesis of Sparse Quantum States}

\author{Renaud Vilmart}{Université Paris-Saclay, Inria, CNRS, ENS Paris-Saclay, Laboratoire Méthodes Formelles, 91190, Gif-sur-Yvette, France.}{renaud.vilmart@inria.fr}{https://orcid.org/0000-0002-8828-4671}{}

\author{Sunheang Ty}{IRT SystemX, Gif-sur-Yvette, France.}{sunheang.ty@irt-systemx.fr}{https://orcid.org/0009-0008-4154-7017}{}

\author{Chetra Mang}{IRT SystemX, Gif-sur-Yvette, France.}{chetra.mang@irt-systemx.fr}{https://orcid.org/0000-0002-0125-2707}{}

\authorrunning{R.~Vilmart \& S.~Ty \& C.~Mang}



\begin{document}

\maketitle

\begin{abstract}
Preparing a quantum circuit that implements a given sparse state is an important building block that is necessary for many different quantum algorithms. In the context of fault-tolerant quantum computing, the so-called non-Clifford gates are much more expensive to perform than the Clifford ones. We hence provide an algorithm for synthesizing sparse quantum states with a special care for quantum resources. The circuit depth, ancilla count, and crucially non-Clifford count of the circuit produced by the algorithm are all linear in the sparsity when access to arbitrary-angled rotations is given. When compiled down to the standard Clifford+T gate set, several constructions can be given for increasingly better T-count and depth at the expense of a larger number of ancillae. The most optimised construction for T-count reaches $\mathcal O\left(s+\log_2(1/\epsilon)\right)$ T gates for error $\epsilon$.

The constructions are broken into two parts, one that synthesises a generalized W-state, well studied in the literature; and the second which is a classical reversible circuit implementing a permutation that maps the basis states of the W-state to those of the target sparse quantum state. We reduce this problem to the diagonalization of a binary matrix, using a specific set of elementary matrix operations corresponding to the classical reversible gates. We then solve this problem using a new version of Gauss-Jordan elimination, that minimizes the circuit complexities including circuit depth using parallel elimination steps. When the circuit is applied in one direction, we notice that all occurrences of (the expensive) Toffoli gates can all be replaced by adaptive Clifford circuits, leading to a better non-Clifford count.
\end{abstract}

\section{Introduction}
\label{sec:intro}

Quantum state preparation, i.e., the procedure to prepare an arbitrary quantum state, is an important algorithmic subroutine in quantum computing and serves as a key component to many quantum algorithms, including Hamiltonian simulations \cite{low2019hamiltonian,berry2015simulating,berry2015hamiltonian,berry2009black,low2017optimal}, quantum linear solvers \cite{harrow2009quantum,ambainis2010variable,childs2017quantum,gilyen2019quantum,an2022quantum,lin2020near,costa2022optimal,dalzell2024shortcut,low2024quantum}, and quantum machine learning \cite{wiebe2012quantum,brandao2017quantum,schuld2021machine,chakraborty2018power,biamonte2017quantum}. Consequently, it plays an essential role in the reasoning and analysis of these algorithms. In spite of this, the lower-bound on the circuit size, (i.e., the total number of quantum gate in a quantum circuit), for preparing an arbitrary $n$-qubit quantum state is $\mathbf{\Omega}(2^n)$ \cite{plesch2011quantum}, i.e., exponential in the number of qubits. In other word, despite its importance, the subroutine turns out to be very resource-inefficient, and likewise has a great impact on the performance of the quantum algorithms which use the subroutine. Several implementations \cite{Zhang2022Quantum,sun2023asymptotically,yuan2023optimal,gui2024spacetime} attest to this fact. For many applications, however, the quantum state needed to be prepared is sparse, i.e., its corresponding vector has only few non-zero entries. In such cases, the computational resources required to synthesize the quantum state are quite efficient as demonstrated in \cite{Zhang2022Quantum,Gleinig2021Efficient,Fomichev2024Initial,Li2025Nearly}. For an $n$-qubit and $s$-sparse quantum state, the quantum circuit for preparing such state has a circuit size as low as $\mathcal{O}(\frac{sn}{log(n)}+n)$ \cite{Li2025Nearly}. In fault-tolerant quantum computing, some gates (so-called Clifford gates) are easy to implement, while gates outside this set require considerably more resources. Non-Clifford count hence becomes an important metric to be minimised, and which has surprisingly seen little consideration in the task of sparse state synthesis.

In this article, we propose a novel methodology to construct a quantum circuit for preparing a sparse quantum state, which is resource-efficient across many metrics. More specifically, our quantum circuit has $\mathcal{O}(s\max\{s,n\})$ size, $\mathcal{O}(s+\max\{s,n\})$ depth, $\mathcal{O}(s)$ non-Clifford count, and $\max\{s,n\}-n$ ancillary qubits. A distinctive feature of our circuit is its depth and non-Clifford count; it is the first to reach a linear non-Clifford count. The depth is (additively) linear in both sparsity $s$ and number of qubit $n$, while the non-Clifford count is linear in sparsity $s$, while being independent on the number of qubit. This is achieved at the cost of slightly higher circuit size and ancillary qubit. In many practical applications, it also often happens that $s\leq n$, in which case no ancilla is needed. The methodology is divided into two parts: (1) a quantum circuit for preparing a generalization of a special quantum state called W-state (which we can already find in \cite{Johri2021Nearest}), and (2) a classically reversible quantum circuit which transforms the terms of the W-state to those of the sparse quantum state. The former comprises solely $CX$ and $CR_Y$ gates, while the latter comprises $X$, $CX$ and $CCX$ gates. In both cases, we primarily attempt to reduce the number of non-Clifford gates, i.e., $CR_Y$ and $CCX$, while also optimizing gate parallelization, which leads to a linear depth. 

Other approaches for synthesizing sparse quantum states are considered in the literature. \Cref{tab:comparison} gives an overview of the different metrics associated with these approaches. Notice that we focus on exact synthesis, and leave out approaches that produce an approximation of the state such as \cite{Feniou2024Sparse}. Interestingly, to the best of our knowledge, non-Clifford count was only considered and optimized in \cite{Fomichev2024Initial,Gleinig2021Efficient}, where the obtained metric is $\mathcal O(s\log(s))$ assuming $s>n$. In both cases, the depth is $\mathcal O(sn)$. Two approaches \cite{Li2025Nearly,Zhang2022Quantum} manage to get logarithmic depth, but at the cost of a large number of ancillas, and a priori of non-Clifford gates.

\begin{table}[!ht]
	\setlength{\tabcolsep}{.25cm}
	\centering
	\begin{tabular}{@{}lllll@{}}
		\toprule
		\textbf{Reference} & 
		\textbf{Size} & 
		\textbf{Depth} & 
		\textbf{Non-Clifford} & 
		\textbf{Ancillas} 
		\\
		\midrule
		\textbf{Ours} & 
		$\mathcal O(s\max\{s,n\})$ & 
		$\mathcal O(s+\max\{s,n\})$ & 
		$\mathcal O(s)$ & 
		$\max\{s,n\}-n$ 
		\\
		Gleining et al.~\cite{Gleinig2021Efficient} & 
		$\mathcal O(s n)$ & 
		$\mathcal O(s n)$ & 
		$\mathcal O(s\log(s)+n)$ & 
		$0$ 
		\rule{0pt}{.6cm}
		\\
		Fomichev et al.~\cite{Fomichev2024Initial} & 
		$\mathcal O(sn)^*$ & 
		$\mathcal O(sn)^*$ & 
		$\mathcal O(s\log(s))$ & 
		$\mathcal O(\log(s))$ 
		\rule{0pt}{.6cm}
		\\
		Li et al.~\cite{Li2025Nearly} & 
		$\mathcal O\left(\frac{s n}{\log(n)}+n\right)$ & 
		$\mathcal O\left(\frac{s n}{\log(n)}+n\right)$ & 
		$\mathcal O\left(\frac{s n}{\log(n)}+n\right)^*$ & 
		$0$
		\rule{0pt}{.6cm}
		\\
		Li et al.~\cite{Li2025Nearly} & 
		$\mathcal O\left(\frac{s n}{\log(s)}\right)$ & 
		$\mathcal O\left(\log(s n)\right)$ & 
		$\mathcal O\left(\frac{s n}{\log(s)}\right)^*$ & 
		$\mathcal O\left(\frac{s n}{\log(s)}\right)$
		\rule{0pt}{.6cm}
		\\
		Zhang et al.~\cite{Zhang2022Quantum} & 
		--- & 
		$\mathcal O\left(\log(s n)\right)$ & 
		--- & 
		$\mathcal O\left(s\log(s)n\right)$
		\rule{0pt}{.6cm}
		\\
		\cite{Malvetti2021Quantum,Ramacciotti2024Simple,Tubman2018Postponing,deVeras2022Double} & 
		$\mathcal O(s n)$ & 
		$\mathcal O(s n)$ & 
		--- & 
		$\mathcal O(1)$
		\rule{0pt}{.6cm}
		\\
		\bottomrule
	\end{tabular}
	\medskip
	\caption{Metrics of different methods for exact sparse state preparation.\\
		--- indicates that the metric was not considered in the referenced paper.\\
		* indicates that the metric was not considered but can be inferred from analysis.}
	\label{tab:comparison}
\end{table}

Many fault-tolerant schemes require a specific gate set, the Clifford+T gate set, where Clifford gates are complemented by a single additional gate, called T gate. Since all states and unitaries can be approximated to arbitrary precision within this gate set, many synthesis results compile down to Clifford+T, and several results about bounds on the number of necessary T gates are known. Recently, a construction that is optimal for arbitrary state preparation has been shown \cite{Gosset2025Preparation}, and yields $\mathcal O\left(\sqrt{2^n\log_2(1/\epsilon)}+\log_2(1/\epsilon)\right)$ T gates to implement $n$-qubit states up to error $\epsilon$. It is possible to use this result for the compilation of sparse states using our construction, yielding a linear T-count $\mathcal O(s+\log_2(1/\epsilon))$. We provide a lower bound of $\mathbf \Omega(\min(s,n))$ T-gates for the synthesis of sparse states, showing that the dependency in $s$ in the aforementioned T-count is optimal when $s=\mathcal O(n)$. This construction, together with a more direct one, yield the metrics gathered in \Cref{tab:resource-ClifT}.
\begin{table}[!ht]
	\setlength{\tabcolsep}{.3cm}
	\centering
	\begin{tabular}{@{}lll@{}}
		\toprule
		\textbf{T-count} & \textbf{Depth} & \textbf{Ancillae} \\
		\midrule
		$\mathcal{O}(s\log_2(s/\epsilon))$ & $\mathcal{O}(s+n+\log_2(s)\log_2(1/\epsilon))$ & $\max(0,s-n)$ \\
		$\mathcal O\left(s+\log_2(1/\epsilon)\right)$ & $\mathcal O\left(s+n+\log_2(1/\epsilon)\right)$ & $\mathcal O\left( s+\log_2(1/\epsilon) \right)$ \\
		\bottomrule
	\end{tabular}
	\medskip
	\caption{Resource analysis for the proposed sparse state preparation methods in fault-tolerant settings.}
	\label{tab:resource-ClifT}
\end{table}

\subsection{Subsequent work}

The subsequent paper \cite{Rupprecht2026Sparse} works by first building a compact state and then redistributing the weights to the target basis states through a permutation (as we do here in the second construction of \cref{thm:sparse-state-prep-ClifT}). The permutation is built in a different manner, and reaches a smaller number of ancillas $\mathcal O(\log_2(s))$ while keeping a quasi-linear T-count of $\mathcal O\left(s+\frac sn\log_2\left(\frac sn\right)\right)$. The focus there is on the small constant factors in this complexity.

Recently, \cite{Luo2026Sparse} extended the current result by:
\begin{enumerate}[i)]
	\item providing an improved $T$-count of $\tilde{\mathcal O}\left(\min\left(s,\sqrt{sn^{\frac32}}\right) +\sqrt{s\log_2(1/\epsilon)} + \log_2(1/\epsilon)\right)$
	\item improving the lower bound to $\mathbf{\Omega}\left(\min\left(s,\sqrt{sn}\right)\right)$ and adapting it to the approximate case.
\end{enumerate}

\subsection{Paper organisation}

The article is organized as follows: \Cref{sec:prelim} provides the necessary background in quantum computing to understand the rest of the paper, \Cref{sec:permut} shows the construction of a permutation circuit that is used in several occurrences, and is optimised in depth and T-count. \Cref{sec:w-state} then shows how to efficiently perform the first part of our construction, both in the ideal setting and in the fault-tolerant setting. Finally, \Cref{sec:sparse-state} combines the results from the previous sections to provide sparse state synthesis algorithms that reach the aforementioned metrics.

\section{Preliminaries}
\label{sec:prelim}

Let us start by introducing the core concepts of quantum computing \cite{nielsen_chuang_2010} that will be needed in this paper.

In the following, we will use the Dirac notation and the matrix representation of the finite-dimensional linear maps interchangeably. $\ket0$ and $\ket1$ denote the two elements of the canonical orthonormal basis of $\mathbb C^2$: $\ket0:=\begin{pmatrix}1&0\end{pmatrix}^\dagger$ and $\ket1:=\begin{pmatrix}0&1\end{pmatrix}^\dagger$, where $\dag$ is the \emph{dagger}, i.e.~the conjugate transpose, or adjoint. A bit string of size $n$ inside the \emph{ket} notation $\ket{...}$ is a shorthand for the Kronecker product of the corresponding kets (e.g.~$\ket{011}:=\ket0\otimes\ket1\otimes\ket1$). Such a state is called an $n$-qubit classical state; and the $2^n$ different $n$-qubit classical states form the canonical orthonormal basis of $\mathbb C^{2^n}$. Any $n$-qubit quantum state $\ket\psi$ is a normalized linear combination of the $n$-qubit classical states. While an unnormalized element of $\mathbb C^{2^n}$ does not constitute a quantum state, they are mathematically relevant and will be used in the following. The \emph{bra} notation $\bra{...}$ represents the dagger of the corresponding ket: $\bra\psi := \ket\psi^\dag$. In the following, we shall use $\ket{e_i}$ to denote the classical state $\ket{0...010...0}$ where the $1$ appears at index $i$. The number of qubits of the state in this notation is ambiguous, but it will usually be clear from the context.

A valid (measurement-free) quantum operation $f$ from $n$ qubits to $m$ qubits has to send any $n$-qubit quantum state to an $m$ qubit quantum state, it is hence an isometry: $f^\dag\circ f=I$, where $I$ is the identity. When $n=m$, this evolution is even a unitary: $f^\dag\circ f=f\circ f^\dag=I$. Notice that any linear map $f:\mathbb C^{2^n}\to\mathbb C^{2^m}$ can be written as either a matrix or a linear combination of compositions of kets and bras. For instance, $I = \ketbra0+\ketbra1$ in $\mathbb C^2$. 

Quantum circuits are visual representations of quantum operations of quantum states in the algorithms. The basic operations, called quantum gates, can be interpreted as liner maps, and composed to form larger circuits which are again interpreted as linear maps. Circuits can be composed either sequentially (by plugging the outputs of the first circuit to the inputs of the second circuit), which interprets as the usual composition of linear maps; or in parallel (by stacking the first on top of the second), which interprets as the Kronecker product of linear maps. 

A universal gate set is a set of quantum gates that can be composed to generate any linear isometry $f:\mathbb C^{2^n}\to\mathbb C^{2^m}$. One such example is called the Clifford+Phase gate set $\{H,CX,P(\alpha),\ket{0}\}$ for all $\alpha \in \mathbb{R}$, whose matrices and visual representations are:
\begin{gather*}
	\begin{array}{l}
		H 
		:= 
		\frac1{\sqrt2}
		\begin{pmatrix}
			1 & 1 \\ 1 & -1
		\end{pmatrix} 
		:= 
		
\begin{tikzpicture}
	\begin{pgfonlayer}{nodelayer}
		\node [style=none] (0) at (-0.5, 0) {};
		\node [style=box] (1) at (0, 0) {$H$};
		\node [style=none] (2) at (0.5, 0) {};
	\end{pgfonlayer}
	\begin{pgfonlayer}{edgelayer}
		\draw (2.center) to (0.center);
	\end{pgfonlayer}
\end{tikzpicture}

		\\[0.7em]
		P(\alpha) 
		:= 
		\begin{pmatrix}
			1&0\\0&e^{i\alpha}
		\end{pmatrix} 
		:= 
		
\begin{tikzpicture}
	\begin{pgfonlayer}{nodelayer}
		\node [style=none] (0) at (-0.75, 0) {};
		\node [style=box] (1) at (0, 0) {$P(\alpha)$};
		\node [style=none] (2) at (0.75, 0) {};
	\end{pgfonlayer}
	\begin{pgfonlayer}{edgelayer}
		\draw (2.center) to (0.center);
	\end{pgfonlayer}
\end{tikzpicture}

	\end{array}
	\qquad\qquad
	CX 
	:= 
	\begin{pmatrix}
		1&0&0&0\\
		0&1&0&0\\
		0&0&0&1\\
		0&0&1&0
	\end{pmatrix} 
	:= 
	
\begin{tikzpicture}
	\begin{pgfonlayer}{nodelayer}
		\node [style=none] (0) at (-0.5, -0.25) {};
		\node [style=none] (1) at (-0.5, 0.25) {};
		\node [style=none] (2) at (0.5, 0.25) {};
		\node [style=none] (3) at (0, -0.25) {$\oplus$};
		\node [style=none] (4) at (0.5, -0.25) {};
		\node [style=dot] (5) at (0, 0.25) {};
	\end{pgfonlayer}
	\begin{pgfonlayer}{edgelayer}
		\draw (1.center) to (2.center);
		\draw (4.center) to (0.center);
		\draw (3.center) to (5);
	\end{pgfonlayer}
\end{tikzpicture}

\end{gather*}
Despite providing universality, $P(\alpha)$ is an infinite and uncountable family of gates, which can pose problems for error correction. Quantum computers are intrinsically noisy, and hence will require an error correction scheme in order to perform most tasks devised for them. 

However, we can transform a circuit in the universal gate set to the Clifford+T gate set $\{H,CX,T,\ket{0}\}$, where $T := P(\pi/4)$. This Clifford+T gate set is said to be approximately universal, i.e., its circuits can approximate any quantum evolution using additional polynomial computational resources \cite{Kitaev1997Quantum,Selinger2015Efficient,Ross2016Optimal}. Another interesting gate set is the Clifford gate set $\{H,CX,S,\ket{0}\}$, where $S := P(\pi/2) = T^2$. This gate set is not universal or approximately universal; however, it can be implemented more efficiently than the Clifford+T set in many error correction schemes. The objective is therefore to first minimize the number of phase gates $P(\alpha)$, and then the number of $T$ gates; both of which are called non-Clifford gates.

It is convenient to define the following (Clifford) gates as well:
\[Z:=P(\pi)=S^2\qquad X:= HZH\qquad 
CZ:=\begin{pmatrix}
	1&0&0&0\\
	0&1&0&0\\
	0&0&1&0\\
	0&0&0&-1
	\end{pmatrix}:=
\begin{tikzpicture}
	\begin{pgfonlayer}{nodelayer}
		\node [style=none] (0) at (0, -0.25) {};
		\node [style=none] (1) at (0, 0.25) {};
		\node [style=none] (2) at (1.75, 0.25) {};
		\node [style=none] (3) at (0.875, -0.25) {$\oplus$};
		\node [style=none] (4) at (1.75, -0.25) {};
		\node [style=dot] (5) at (0.875, 0.25) {};
		\node [style=none] (6) at (-0.5, 0) {$=$};
		\node [style=none] (7) at (-1.75, -0.25) {};
		\node [style=none] (8) at (-1.75, 0.25) {};
		\node [style=none] (9) at (-1, 0.25) {};
		\node [style=dot] (10) at (-1.375, -0.25) {};
		\node [style=none] (11) at (-1, -0.25) {};
		\node [style=dot] (12) at (-1.375, 0.25) {};
		\node [style=box] (13) at (0.375, -0.25) {$H$};
		\node [style=box] (14) at (1.375, -0.25) {$H$};
	\end{pgfonlayer}
	\begin{pgfonlayer}{edgelayer}
		\draw (1.center) to (2.center);
		\draw (4.center) to (0.center);
		\draw (3.center) to (5);
		\draw (8.center) to (9.center);
		\draw (11.center) to (7.center);
		\draw (10) to (12);
	\end{pgfonlayer}
\end{tikzpicture}
\]

In this article, we will encounter three types of non-Clifford gates: $CR_Y(\theta)$, $CH$, and $CCX$ gates. The $CR_Y(\theta)$ gate is a two-qubit gates which can be composed using the Clifford+Phase gate set. A special case (up to a permutation of columns or rows) is $CH$ which can be decomposed using the Clifford+T gate set. Their matrix representations and circuit decomposition are given as follows:
\begin{gather*}
	CR_Y(\theta)
	:=
	\begin{pmatrix}
		I & 0 \\
		0 & R_Y(\theta)
	\end{pmatrix}
	\quad
	\text{where}
	\quad
	R_Y(\theta)
	:=
	\begin{pmatrix}
		\cos(\frac\theta2) & -\sin(\frac\theta2)\\[0.5em]
		\sin(\frac\theta2) & \cos(\frac\theta2)
	\end{pmatrix}
	\\[.5cm]
	
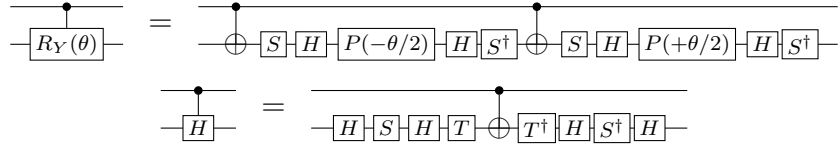
\begin{tikzpicture}
	\begin{pgfonlayer}{nodelayer}
		\node [style=none] (0) at (-4, -0.25) {};
		\node [style=none] (1) at (-4, 0.25) {};
		\node [style=none] (2) at (-2.5, 0.25) {};
		\node [style=box] (3) at (-3.25, -0.25) {$R_Y(\theta)$};
		\node [style=none] (5) at (-2.5, -0.25) {};
		\node [style=dot] (6) at (-3.25, 0.25) {};
		\node [style=none] (7) at (-2, 0) {$=$};
		\node [style=none] (8) at (-1.5, -0.25) {};
		\node [style=none] (9) at (-1.5, 0.25) {};
		\node [style=none] (10) at (7, 0.25) {};
		\node [style=none] (11) at (3, -0.25) {$\oplus$};
		\node [style=none] (12) at (7, -0.25) {};
		\node [style=dot] (13) at (3, 0.25) {};
		\node [style=box] (14) at (-0.5, -0.25) {$S$};
		\node [style=box] (15) at (0, -0.25) {$H$};
		\node [style=box] (16) at (1, -0.25) {$P(-\theta/2)$};
		\node [style=box] (17) at (2, -0.25) {$H$};
		\node [style=box] (18) at (6, -0.25) {$H$};
		\node [style=box] (19) at (5, -0.25) {$P(+\theta/2)$};
		\node [style=box] (20) at (4, -0.25) {$H$};
		\node [style=box] (21) at (3.5, -0.25) {$S$};
		\node [style=none] (22) at (-1, -0.25) {$\oplus$};
		\node [style=dot] (23) at (-1, 0.25) {};
		\node [style=box] (24) at (2.5, -0.25) {$S^\dag$};
		\node [style=box] (25) at (6.5, -0.25) {$S^\dag$};
	\end{pgfonlayer}
	\begin{pgfonlayer}{edgelayer}
		\draw (1.center) to (2.center);
		\draw (5.center) to (0.center);
		\draw (3) to (6);
		\draw (9.center) to (10.center);
		\draw (12.center) to (8.center);
		\draw (11.center) to (13);
		\draw (22.center) to (23);
	\end{pgfonlayer}
\end{tikzpicture}

	\\[.2cm]
	
\begin{tikzpicture}
	\begin{pgfonlayer}{nodelayer}
		\node [style=none] (0) at (-2.5, -0.25) {};
		\node [style=none] (1) at (-2.5, 0.25) {};
		\node [style=none] (2) at (-1.5, 0.25) {};
		\node [style=box] (3) at (-2, -0.25) {$H$};
		\node [style=none] (5) at (-1.5, -0.25) {};
		\node [style=dot] (6) at (-2, 0.25) {};
		\node [style=none] (7) at (-1, 0) {$=$};
		\node [style=none] (8) at (-0.5, -0.25) {};
		\node [style=none] (9) at (-0.5, 0.25) {};
		\node [style=none] (10) at (4.5, 0.25) {};
		\node [style=none] (11) at (2, -0.25) {$\oplus$};
		\node [style=none] (12) at (4.5, -0.25) {};
		\node [style=dot] (13) at (2, 0.25) {};
		\node [style=box] (14) at (0, -0.25) {$H$};
		\node [style=box] (15) at (0.5, -0.25) {$S$};
		\node [style=box] (16) at (1, -0.25) {$H$};
		\node [style=box] (17) at (1.5, -0.25) {$T$};
		\node [style=box] (18) at (4, -0.25) {$H$};
		\node [style=box] (19) at (3.5, -0.25) {$S^\dag$};
		\node [style=box] (20) at (3, -0.25) {$H$};
		\node [style=box] (21) at (2.5, -0.25) {$T^\dag$};
	\end{pgfonlayer}
	\begin{pgfonlayer}{edgelayer}
		\draw (1.center) to (2.center);
		\draw (5.center) to (0.center);
		\draw (3) to (6);
		\draw (9.center) to (10.center);
		\draw (12.center) to (8.center);
		\draw (11.center) to (13);
	\end{pgfonlayer}
\end{tikzpicture}

\end{gather*}
The $CCX$ gate is a three-qubit gates decomposable in the Clifford+T gate set as follows:
\begin{gather*}
	
\begin{tikzpicture}
	\begin{pgfonlayer}{nodelayer}
		\node [style=none] (0) at (-1, -0.5) {};
		\node [style=none] (1) at (-1, 0.5) {};
		\node [style=none] (2) at (-0.25, 0.5) {};
		\node [style=none] (3) at (-0.625, -0.5) {$\oplus$};
		\node [style=none] (4) at (-0.25, -0.5) {};
		\node [style=dot] (5) at (-0.625, 0.5) {};
		\node [style=none] (6) at (-1, 0) {};
		\node [style=none] (7) at (-0.25, 0) {};
		\node [style=dot] (8) at (-0.625, 0) {};
		\node [style=none] (9) at (0.25, 0) {$=$};
		\node [style=none] (10) at (0.75, -0.5) {};
		\node [style=none] (11) at (0.75, 0.5) {};
		\node [style=none] (12) at (7, 0.5) {};
		\node [style=none] (13) at (1.75, -0.5) {$\oplus$};
		\node [style=none] (14) at (7, -0.5) {};
		\node [style=dot] (15) at (1.75, 0) {};
		\node [style=none] (16) at (0.75, 0) {};
		\node [style=none] (17) at (7, 0) {};
		\node [style=box] (19) at (1.25, -0.5) {$H$};
		\node [style=none] (20) at (2.75, -0.5) {$\oplus$};
		\node [style=dot] (21) at (2.75, 0.5) {};
		\node [style=box] (22) at (2.25, -0.5) {$T^\dag$};
		\node [style=box] (23) at (3.25, -0.5) {$T$};
		\node [style=none] (24) at (3.75, -0.5) {$\oplus$};
		\node [style=dot] (25) at (3.75, 0) {};
		\node [style=none] (26) at (4.75, -0.5) {$\oplus$};
		\node [style=dot] (27) at (4.75, 0.5) {};
		\node [style=box] (28) at (4.25, -0.5) {$T^\dag$};
		\node [style=box] (29) at (5.25, -0.5) {$T$};
		\node [style=box] (30) at (5.75, -0.5) {$H$};
		\node [style=box] (31) at (5.25, 0) {$T$};
		\node [style=none] (32) at (5.75, 0) {$\oplus$};
		\node [style=dot] (33) at (5.75, 0.5) {};
		\node [style=box] (34) at (6.25, 0) {$T^\dag$};
		\node [style=box] (35) at (6.25, 0.5) {$T$};
		\node [style=none] (36) at (6.75, 0) {$\oplus$};
		\node [style=dot] (37) at (6.75, 0.5) {};
	\end{pgfonlayer}
	\begin{pgfonlayer}{edgelayer}
		\draw (1.center) to (2.center);
		\draw (4.center) to (0.center);
		\draw (3.center) to (5);
		\draw (6.center) to (7.center);
		\draw (11.center) to (12.center);
		\draw (14.center) to (10.center);
		\draw (13.center) to (15);
		\draw (16.center) to (17.center);
		\draw (20.center) to (21);
		\draw (24.center) to (25);
		\draw (26.center) to (27);
		\draw (32.center) to (33);
		\draw (36.center) to (37);
	\end{pgfonlayer}
\end{tikzpicture}

\end{gather*}
When the target state is initially in $\ket{0}$, or if we allow an ancilla, the number of T-gates can be reduced to $4$ \cite{Jones2013LowOverhead}. A more important special case where the $CCX$ gate can be simplified, is when the last qubit is guaranteed to be set to $\ket0$ after its application. In that case, the gate can be implemented in the Clifford fragment, using a measurement and a classically controlled $CZ$ gate \cite{Jones2013LowOverhead,Gidney2018halving,Beverland2020Lower}:
\[
\begin{tikzpicture}
	\begin{pgfonlayer}{nodelayer}
		\node [style=none] (0) at (-0.5, -0.5) {};
		\node [style=none] (1) at (-0.5, 0.5) {};
		\node [style=none] (2) at (1, 0.5) {};
		\node [style=none] (3) at (0, -0.5) {$\oplus$};
		\node [style=none] (4) at (0.5, -0.5) {};
		\node [style=dot] (5) at (0, 0.5) {};
		\node [style=none] (6) at (-0.5, 0) {};
		\node [style=none] (7) at (1, 0) {};
		\node [style=dot] (8) at (0, 0) {};
		\node [style=none] (9) at (0.75, -0.5) {$\ket0$};
	\end{pgfonlayer}
	\begin{pgfonlayer}{edgelayer}
		\draw (1.center) to (2.center);
		\draw (4.center) to (0.center);
		\draw (3.center) to (5);
		\draw (6.center) to (7.center);
	\end{pgfonlayer}
\end{tikzpicture}
~~=~~
\begin{tikzpicture}
	\begin{pgfonlayer}{nodelayer}
		\node [style=none] (0) at (-2, -0.5) {};
		\node [style=none] (1) at (-2, 0.5) {};
		\node [style=none] (2) at (-1, 0.5) {};
		\node [style=none] (4) at (-1.5, -0.5) {};
		\node [style=dot] (5) at (-1.5, 0.5) {};
		\node [style=none] (6) at (-2, 0) {};
		\node [style=none] (7) at (-1, 0) {};
		\node [style=dot] (8) at (-1.5, 0) {};
		\node [style=none] (9) at (-0.25, 0) {$=$};
		\node [style=none] (10) at (0.5, -0.5) {};
		\node [style=none] (11) at (0.5, 0.5) {};
		\node [style=none] (12) at (2.75, 0.5) {};
		\node [style=dot] (14) at (2.25, 0.5) {};
		\node [style=none] (15) at (0.5, 0) {};
		\node [style=none] (16) at (2.75, 0) {};
		\node [style=dot] (17) at (2.25, 0) {};
		\node [style=none] (18) at (1.325, -0.5) {};
		\node [style=none] (19) at (1.325, -0.325) {};
		\node [style=none] (20) at (1.325, -0.675) {};
		\node [style=none] (21) at (1.925, -0.325) {};
		\node [style=none] (22) at (1.925, -0.675) {};
		\node [style=none] (23) at (1.4, -0.6) {};
		\node [style=none] (24) at (1.85, -0.6) {};
		\node [style=none] (25) at (1.625, -0.6) {};
		\node [style=none] (26) at (1.75, -0.4) {};
		\node [style=none] (27) at (1.925, -0.475) {};
		\node [style=none] (28) at (2.225, -0.475) {};
		\node [style=none] (29) at (1.925, -0.525) {};
		\node [style=none] (30) at (2.275, -0.525) {};
		\node [style=none] (32) at (2.225, 0.05) {};
		\node [style=none] (33) at (2.275, -0.025) {};
		\node [style=box] (34) at (0.875, -0.5) {$H$};
	\end{pgfonlayer}
	\begin{pgfonlayer}{edgelayer}
		\draw (1.center) to (2.center);
		\draw (4.center) to (0.center);
		\draw (6.center) to (7.center);
		\draw (4.center) to (5);
		\draw (11.center) to (12.center);
		\draw (15.center) to (16.center);
		\draw (20.center) to (22.center);
		\draw (22.center) to (21.center);
		\draw (21.center) to (19.center);
		\draw (19.center) to (20.center);
		\draw [bend left=60] (23.center) to (24.center);
		\draw [style={arrows={->[]}}] (25.center) to (26.center);
		\draw (27.center) to (28.center);
		\draw (29.center) to (30.center);
		\draw (10.center) to (18.center);
		\draw (14) to (17);
		\draw (30.center) to (33.center);
		\draw (32.center) to (28.center);
	\end{pgfonlayer}
\end{tikzpicture}
\]

\section{Partial Permutation}
\label{sec:permut}

The main tool that we will use in the following is a construction for a circuit that maps each $\ket{e_i}$ to a target classical state.

\subsection{New construction}

We propose here a construction for such a circuit that turns out to be efficient in depth and non-Clifford count. More specifically:

\begin{theorem}
	\label{thm:permut}
	Let $\chi = (\chi_i)\in\left(\{0,1\}^n\right)^s$ be a sequence of distinct length-$n$ bit strings. Let $m:=\max(s,n)$. There exists a unitary $U_\chi$ such that $U_\chi:\ket{e_i}\otimes\ket{0^{\otimes m-s}}\mapsto\ket{\chi_i}\otimes\ket{0^{\otimes m-n}}$ and which can be implemented by a circuit with $\mathcal O(s)$ T-count,  $\mathcal O(s+n)$ depth, and $\mathcal O(s(s+n))$ gate-count.
	
	There exists a Clifford circuit $C_\chi$ with qubit initialisation, measurements and classically-controlled $CZ$-gates, that maps each $\ket{e_i}$ to $\ket{\chi_i}$ with $\mathcal O(s+n)$ depth, and $\mathcal O(s(s+n))$ gate-count.
\end{theorem}
As we will see in the following, the circuit can actually be built with at most $s-\log_2(s)$ Toffolis, giving a more precise and low-constant T-count.

The second construction makes use of intermediary measurements to entirely rid us of the need for T gates. The presence of intermediary measurements means however that the circuit cannot be easily transposed to map the $\ket{\chi_i}$ to $\ket{e_i}$. On the other hand, the first construction, being a unitary with no measurement, can be transposed right away. The rest of this section is devoted to explaining the construction for $U_\chi$, with the circuit $C_\chi$ being a fairly direct consequence.

We work our way from the target states $\chi_i$ and show how to build states $\ket{e_i}$ out of them. Let $s$, $n$, $m$ and $\chi$ be defined as in the theorem. Build the $m\times s$ matrix $M_\chi$ whose columns are the $\ket{\chi_i}\otimes\ket{0^{\otimes m-n}}$, i.e.~the target states padded with $0$s if $s>n$. The number of rows of $M_\chi$ is hence larger or equal to its number of columns. Notice that the columns of $M_\chi$ are all distinct since all $\chi_i$ are distinct. Imagine the columns represent the classical states in the current state $\ket\psi$. Then applying classical reversible gates ($X$, $CX$ and $CCX$) to the state will amount to modifying the rows of the matrix, in the following way:
\begin{center}
	\begin{tabular}{@{}p{.45\textwidth}p{.33\textwidth}@{}}
		\toprule
		$X$ on qubit $i$ &
		$\tt row_i := row_i \oplus 1$ \\
		\midrule
		$CX$ from qubit $i$ to qubit $j$ &
		$\tt row_j := row_i \oplus row_j$ \\
		\midrule
		$CCX$ from qubits $i,j$ to qubit $k$ &
		$\tt row_k := row_k \oplus row_i \cdot row_j$ \\
		\bottomrule
	\end{tabular}
\end{center}
where $\cdot$ and $\oplus$ denote entry-wise row multiplication and modulo-$2$ addition respectively, and $\tt 1$ denotes a row with $1$ entries. Note that all these operations preserve the fact that all columns are distinct. We can also use $X$ to ensure $M_\chi$ has no $0$-column, by choosing an appropriate qubit (row) to apply it to. Such row can be found as follows: notice that any column with weight $\geq2$ cannot be turned into a $0$ column by application of $X$ on a single of its elements. Hence, picking an index that is not the index of the single non-zero element of a weight-$1$ column yields an appropriate qubit to apply $X$ to. Such index necessarily exists since $m\geq s$. In the following, we assume this has been dealt with and that $M_\chi$ has no zero-column.

The goal is then to turn the matrix into the identity matrix (or the unit matrix $\begin{pmatrix}I\\0\end{pmatrix}$ when $m>s$). To reduce the depth of the resulting circuit, we want to parallelise the gates as much as possible. This is made possible by first computing the PLUQ decomposition of matrix $M_\chi$:
\[M_\chi = PLUQ = P\underbrace{\left(
\begin{tikzpicture}
	\begin{pgfonlayer}{nodelayer}
		\node [style=none] (0) at (-0.75, 1.25) {$1$};
		\node [style=none] (1) at (1.25, -0.75) {$1$};
		\node [style=none] (3) at (-0.5, 1) {};
		\node [style=none] (4) at (1, -0.5) {};
		\node [style=none] (5) at (0.125, 0.125) {};
		\node [style=none] (6) at (-0.75, -0.75) {};
		\node [style=none] (11) at (-0.125, 0.375) {};
		\node [style=none] (12) at (-0.75, -0.25) {};
		\node [style=none] (13) at (-0.375, 0.625) {};
		\node [style=none] (14) at (-0.75, 0.25) {};
		\node [style=none] (16) at (-0.75, 1) {$\cdot$};
		\node [style=none] (22) at (0, 0.25) {};
		\node [style=none] (23) at (-0.75, -0.5) {};
		\node [style=none] (24) at (-0.25, 0.5) {};
		\node [style=none] (25) at (-0.75, 0) {};
		\node [style=none] (26) at (-0.5, 0.75) {};
		\node [style=none] (27) at (-0.75, 0.5) {};
		\node [style=none] (34) at (-0.625, 0.875) {};
		\node [style=none] (35) at (-0.75, 0.75) {};
		\node [style=none] (38) at (-0.75, -1) {};
		\node [style=none] (39) at (-0.75, -1.25) {};
		\node [style=none] (40) at (-0.75, -1.5) {};
		\node [style=none] (44) at (-0.5, -1.5) {};
		\node [style=none] (45) at (-0.25, -1.5) {};
		\node [style=none] (46) at (0, -1.5) {};
		\node [style=none] (47) at (0.25, -1.5) {};
		\node [style=none] (48) at (0.25, -1.25) {};
		\node [style=none] (49) at (0.25, -1) {};
		\node [style=none] (50) at (0.25, -0.75) {};
		\node [style=none] (51) at (0.25, -0.5) {};
		\node [style=none] (52) at (0.25, -0.25) {};
		\node [style=none] (53) at (0.25, 0) {};
		\node [style=none] (54) at (0.25, 1.25) {};
	\end{pgfonlayer}
	\begin{pgfonlayer}{edgelayer}
		\draw (4.center) to (3.center);
		\draw (5.center) to (6.center);
		\draw (11.center) to (12.center);
		\draw (13.center) to (14.center);
		\draw (22.center) to (23.center);
		\draw (24.center) to (25.center);
		\draw (26.center) to (27.center);
		\draw (34.center) to (35.center);
		\draw (53.center) to (38.center);
		\draw (52.center) to (39.center);
		\draw (51.center) to (40.center);
		\draw (50.center) to (44.center);
		\draw (49.center) to (45.center);
		\draw (48.center) to (46.center);
		\draw [style={dashed}, draw=gray] (54.center) to (47.center);
	\end{pgfonlayer}
\end{tikzpicture}
\right)}_{L}
\underbrace{\left(
\begin{tikzpicture}
	\begin{pgfonlayer}{nodelayer}
		\node [style=none] (0) at (-1.25, 1) {$1$};
		\node [style=none] (1) at (-0.25, 0) {$1$};
		\node [style=none] (3) at (-1, 0.75) {};
		\node [style=none] (4) at (-0.5, 0.25) {};
		\node [style=none] (5) at (-0.125, 0.125) {};
		\node [style=none] (6) at (0.75, 1) {};
		\node [style=none] (11) at (-0.375, 0.375) {};
		\node [style=none] (12) at (0.25, 1) {};
		\node [style=none] (13) at (-0.625, 0.625) {};
		\node [style=none] (14) at (-0.25, 1) {};
		\node [style=none] (16) at (-1, 1) {$\cdot$};
		\node [style=none] (18) at (0.75, 0.5) {};
		\node [style=none] (20) at (0.75, 0) {};
		\node [style=none] (22) at (-0.25, 0.25) {};
		\node [style=none] (23) at (0.5, 1) {};
		\node [style=none] (24) at (-0.5, 0.5) {};
		\node [style=none] (25) at (0, 1) {};
		\node [style=none] (26) at (-0.75, 0.75) {};
		\node [style=none] (27) at (-0.5, 1) {};
		\node [style=none] (28) at (0, 0) {};
		\node [style=none] (29) at (0.75, 0.75) {};
		\node [style=none] (31) at (0.75, 0.25) {};
		\node [style=none] (34) at (-0.875, 0.875) {};
		\node [style=none] (35) at (-0.75, 1) {};
		\node [style=none] (36) at (0.25, 0) {};
		\node [style=none] (37) at (0.5, 0) {};
		\node [style=none] (38) at (-1.25, 0) {};
		\node [style=none] (39) at (-0.25, -1) {};
	\end{pgfonlayer}
	\begin{pgfonlayer}{edgelayer}
		\draw (4.center) to (3.center);
		\draw (5.center) to (6.center);
		\draw (11.center) to (12.center);
		\draw (13.center) to (14.center);
		\draw (22.center) to (23.center);
		\draw (24.center) to (25.center);
		\draw (26.center) to (27.center);
		\draw (28.center) to (29.center);
		\draw (34.center) to (35.center);
		\draw (18.center) to (36.center);
		\draw (31.center) to (37.center);
		\draw [style={dashed}, draw=gray] (38.center) to (20.center);
		\draw [style={dashed}, draw=gray] (14.center) to (39.center);
	\end{pgfonlayer}
\end{tikzpicture}
\right)}_{U}Q\]
where $L$ is a unit $m\times s$ full-rank lower trapezoidal matrix, $U$ is an $s\times s$ upper triangular matrix of rank $r=\rank(M_\chi)$, and $P$ and $Q$ are permutation matrices. The presence of the two latter matrices allows us to flush all pivots of $U$ to the top left as shown above. Once the circuit induced by $L$ and $U$ is built permutation $P$ will amount to a permutation of output wires, which can be pushed through the circuit to meet the other permutation $Q$, a permutation of the input qubits (which will be modified when building the circuit for $U$). When applying the construction to a state, the permutation can be pushed to the state preparation. Otherwise, it can be built using a $CX$ circuit of depth $6$ \cite{Moore2001Parallel}, which has no effect on the asymptotic complexity in depth, Clifford--count, non-Clifford-count or T-count.

With the LU decomposition above, we can defer the goal of turning the matrix into the unit matrix to matrices $L$ and $U$. Notice that if the initial matrix is full rank, it suffices to apply Gaussian elimination to turn the matrices into unit ones, which can be done using only $CX$ gates. The $X$ and $CCX$ gates are hence only required to augment the rank of the matrix.
\smallskip

\noindent
\smallskip
\begin{minipage}{0.7\columnwidth}
	\indent $L$ being full rank, we can use it to illustrate how to arrange the $CX$ to get linear depth. The idea is to group the non-diagonal entries of $L$ into a collection of $m+r\leq \mathcal O(n+s)$ anti-diagonals. Each of these anti-diagonals can be removed at once by a depth-$1$ layer of $CX$s, as illustrated on the right, where each arrow represents a potential application of $CX$ with its tail being the source qubit and its head being the target qubit.
\end{minipage}
\hfill
\begin{minipage}{0.22\columnwidth}
	$\left(
\begin{tikzpicture}
	\begin{pgfonlayer}{nodelayer}
		\node [style=none] (0) at (-0.75, 1.25) {$1$};
		\node [style=none] (1) at (1.25, -0.75) {$1$};
		\node [style=none] (3) at (-0.5, 1) {};
		\node [style=none] (4) at (1, -0.5) {};
		\node [style=none] (16) at (-0.75, 1) {};
		\node [style=none] (17) at (-0.5, -1.5) {};
		\node [style=none] (30) at (0.5, -0.25) {};
		\node [style=none] (34) at (1.25, -1.25) {};
		\node [style=none] (35) at (1, -1.5) {};
		\node [style=none] (39) at (0.625, -0.375) {};
		\node [style=none] (40) at (-0.75, -1.5) {};
		\node [style=none] (41) at (0.25, -0.5) {};
		\node [style=none] (42) at (-0.5, -1.25) {};
		\node [style=none] (43) at (-0.25, -1) {};
		\node [style=none] (44) at (0, -0.75) {};
		\node [style=none] (45) at (0.5, 0) {};
		\node [style=none] (46) at (0.25, 0.25) {};
		\node [style=none] (47) at (-0.25, 0.75) {};
		\node [style=none] (48) at (0, 0.5) {};
		\node [style=none] (49) at (-0.25, -1.5) {};
		\node [style=none] (50) at (0.75, -0.5) {};
		\node [style=none] (51) at (0, -1.5) {};
		\node [style=none] (52) at (0.875, -0.625) {};
		\node [style=none] (53) at (0.25, -1.5) {};
		\node [style=none] (54) at (1, -0.75) {};
		\node [style=none] (55) at (0.5, -1.5) {};
		\node [style=none] (56) at (1.125, -0.875) {};
		\node [style=none] (57) at (0.75, -1.5) {};
		\node [style=none] (58) at (1.25, -1) {};
		\node [style=none] (59) at (1.25, -1.5) {$\cdot$};
	\end{pgfonlayer}
	\begin{pgfonlayer}{edgelayer}
		\draw (4.center) to (3.center);
		\draw (34.center) to (35.center);
		\draw (40.center) to (30.center);
		\draw (39.center) to (17.center);
		\draw [style={arrows={->[]}}, draw=gray] (45.center) to (30.center);
		\draw [style={arrows={->[]}}, draw=gray] (46.center) to (41.center);
		\draw [style={arrows={->[]}}, draw=gray] (3.center) to (42.center);
		\draw [style={arrows={->[]}}, draw=gray] (47.center) to (43.center);
		\draw [style={arrows={->[]}}, draw=gray] (16.center) to (40.center);
		\draw [style={arrows={->[]}}, draw=gray] (48.center) to (44.center);
		\draw (50.center) to (49.center);
		\draw (52.center) to (51.center);
		\draw (54.center) to (53.center);
		\draw (56.center) to (55.center);
		\draw (58.center) to (57.center);
	\end{pgfonlayer}
\end{tikzpicture}
\right)$
\end{minipage}\\
The $CX$ is only applied if their is a $1$ at the arrow's head. Doing so for each anti-diagonal obviously yields a $CX$ circuit of depth $\mathcal O(s+n)$ and $\mathcal O(s(s+n))$ gate-count. The algorithm is given in more details in \Cref{alg:anti-diag-rem}, in the case of an \emph{upper} triangular matrix. This process obviously cannot be applied to turn $U$ into the identity in general, since it is not full-rank. It however can serve as a basis, which will have to be interleaved with applications of $CCX$ gates to augment its rank. 

To deal with $U$, we divide our transformation into $s-r$ iterations, as follows:
\begin{gather*}
	\underbrace{\left(
\begin{tikzpicture}
	\begin{pgfonlayer}{nodelayer}
		\node [style=none] (0) at (-1.25, 1) {$1$};
		\node [style=none] (1) at (-0.25, 0) {$1$};
		\node [style=none] (3) at (-1, 0.75) {};
		\node [style=none] (4) at (-0.5, 0.25) {};
		\node [style=none] (5) at (-0.125, 0.125) {};
		\node [style=none] (6) at (0.5, 0.75) {};
		\node [style=none] (11) at (-0.375, 0.375) {};
		\node [style=none] (12) at (0, 0.75) {};
		\node [style=none] (14) at (-0.25, 1) {};
		\node [style=none] (18) at (0.75, 0.5) {};
		\node [style=none] (20) at (0.75, 0) {};
		\node [style=none] (22) at (-0.25, 0.25) {};
		\node [style=none] (23) at (0.25, 0.75) {};
		\node [style=none] (24) at (-0.5, 0.5) {};
		\node [style=none] (25) at (-0.25, 0.75) {};
		\node [style=none] (28) at (0, 0) {};
		\node [style=none] (29) at (0.75, 0.75) {};
		\node [style=none] (31) at (0.75, 0.25) {};
		\node [style=none] (36) at (0.25, 0) {};
		\node [style=none] (37) at (0.5, 0) {};
		\node [style=none] (38) at (-1.25, 0) {};
		\node [style=none] (39) at (-0.25, -1) {};
		\node [style=none] (40) at (0.75, 0.75) {};
		\node [style=none] (41) at (-1.25, 0.75) {};
		\node [style=none] (42) at (0.25, 1) {$0$};
		\node [style=none] (43) at (0.75, 0) {$\cdot$};
	\end{pgfonlayer}
	\begin{pgfonlayer}{edgelayer}
		\draw (4.center) to (3.center);
		\draw (5.center) to (6.center);
		\draw (11.center) to (12.center);
		\draw (22.center) to (23.center);
		\draw (24.center) to (25.center);
		\draw (28.center) to (29.center);
		\draw (18.center) to (36.center);
		\draw (31.center) to (37.center);
		\draw [style=dashed, draw=gray] (38.center) to (20.center);
		\draw [style=dashed, draw=gray] (14.center) to (39.center);
		\draw [style=dashed, draw=gray] (41.center) to (40.center);
	\end{pgfonlayer}
\end{tikzpicture}
\right)}_{U_r}
	\raisebox{2em}{\scriptsize \!\!$j$}
	\mapsto
	\underbrace{\left(
\begin{tikzpicture}
	\begin{pgfonlayer}{nodelayer}
		\node [style=none] (0) at (-1.25, 1) {$1$};
		\node [style=none] (1) at (0, -0.25) {$1$};
		\node [style=none] (3) at (-1, 0.75) {};
		\node [style=none] (4) at (-0.25, 0) {};
		\node [style=none] (5) at (-0.125, 0.125) {};
		\node [style=none] (6) at (0.5, 0.75) {};
		\node [style=none] (11) at (-0.375, 0.375) {};
		\node [style=none] (12) at (0, 0.75) {};
		\node [style=none] (14) at (-0.25, 1) {};
		\node [style=none] (17) at (0.125, -0.125) {};
		\node [style=none] (18) at (0.75, 0.5) {};
		\node [style=none] (20) at (0.75, 0) {};
		\node [style=none] (22) at (-0.25, 0.25) {};
		\node [style=none] (23) at (0.25, 0.75) {};
		\node [style=none] (25) at (0.25, 1) {$0$};
		\node [style=none] (28) at (0, 0) {};
		\node [style=none] (29) at (0.75, 0.75) {};
		\node [style=none] (30) at (0.25, -0.25) {};
		\node [style=none] (31) at (0.75, 0.25) {};
		\node [style=none] (38) at (-0.25, -1) {};
		\node [style=none] (39) at (-1.25, 0) {};
		\node [style=none] (40) at (0.5, -0.25) {};
		\node [style=none] (41) at (0.75, 0.75) {};
		\node [style=none] (42) at (-1.25, 0.75) {};
		\node [style=none] (43) at (0.75, -0.25) {$\cdot$};
	\end{pgfonlayer}
	\begin{pgfonlayer}{edgelayer}
		\draw (4.center) to (3.center);
		\draw (5.center) to (6.center);
		\draw (11.center) to (12.center);
		\draw (17.center) to (18.center);
		\draw (22.center) to (23.center);
		\draw (28.center) to (29.center);
		\draw (30.center) to (31.center);
		\draw [style=dashed, draw=gray] (38.center) to (14.center);
		\draw [style=dashed, draw=gray] (39.center) to (20.center);
		\draw (20.center) to (40.center);
		\draw [style=dashed, draw=gray] (42.center) to (41.center);
	\end{pgfonlayer}
\end{tikzpicture}
\right)}_{U_{r+1}}
	\mapsto
	\cdots
	\mapsto
	\underbrace{\left(
\begin{tikzpicture}
	\begin{pgfonlayer}{nodelayer}
		\node [style=none] (0) at (-1.25, 1) {$1$};
		\node [style=none] (1) at (0.75, -1) {$1$};
		\node [style=none] (3) at (-1, 0.75) {};
		\node [style=none] (4) at (0.5, -0.75) {};
		\node [style=none] (6) at (0.75, 1) {$0$};
		\node [style=none] (14) at (-0.25, 1) {};
		\node [style=none] (17) at (0.125, -0.125) {};
		\node [style=none] (18) at (0.75, 0.5) {};
		\node [style=none] (19) at (0.375, -0.375) {};
		\node [style=none] (20) at (0.75, 0) {};
		\node [style=none] (21) at (0.75, -0.75) {$\cdot$};
		\node [style=none] (28) at (0, 0) {};
		\node [style=none] (29) at (0.75, 0.75) {};
		\node [style=none] (30) at (0.25, -0.25) {};
		\node [style=none] (31) at (0.75, 0.25) {};
		\node [style=none] (32) at (0.5, -0.5) {};
		\node [style=none] (33) at (0.75, -0.25) {};
		\node [style=none] (36) at (0.625, -0.625) {};
		\node [style=none] (37) at (0.75, -0.5) {};
		\node [style=none] (38) at (-1.25, 0) {};
		\node [style=none] (39) at (-0.25, -1) {};
	\end{pgfonlayer}
	\begin{pgfonlayer}{edgelayer}
		\draw (4.center) to (3.center);
		\draw (17.center) to (18.center);
		\draw (19.center) to (20.center);
		\draw (28.center) to (29.center);
		\draw (30.center) to (31.center);
		\draw (32.center) to (33.center);
		\draw (36.center) to (37.center);
		\draw [style={dashed}, draw=gray] (14.center) to (39.center);
		\draw [style={dashed}, draw=gray] (38.center) to (20.center);
	\end{pgfonlayer}
\end{tikzpicture}
\right)}_{U_s}
\end{gather*}
First, using solely $CX$ gates, we can remove the anti-diagonals of $U$,  as long as the whole anti-diagonal is on the left of the vertical dashed line (materialising the rank $r$). While removing the anti-diagonals, it is possible that the first few rows have weight $1$ (with the $1$s on the diagonal). The first row with weight $\geq2$ is important. Let $j$ be its index (if it exists, $s+1$ otherwise)\footnote{Notice that it necessarily exists if $U$ is rank-deficient.}, and suppose it is kept updated throughout the algorithm. Notice that $j<r$ unless the matrix is already the identity. One can check that, with the following algorithm, $j$ can only increase. To get to $U_r$, we remove all the anti-diagonals starting from the left, until we reach the one that goes through index $(j,r+1)$. Then, for each iteration step $t \in \{r,\ldots,s\}$, we perform the following:
\begin{enumerate}
	\item Find a column (on the right of $t$) with entry $1$ in row $j$. Such a column necessarily has weight $\geq2$, otherwise it would either be a zero column, or equal to column $j$, which is prevented by the invariants of the matrix. There hence exists another row index $j_t\neq j$ whose entry is $1$ in the column. Permute this column with the $(t+1)$\textsuperscript{th} column, using some $s \times s$ permutation $Q_t$. Column $t+1$ is now the first (when scanned left to right) to have $1$ entries at both $j$ and $j_t$. 
	\item Apply a $CCX$ gate from rows $j$ and $j_t$, to row $(t+1)$. This transforms the $(t+1)$\textsuperscript{th} row, which was previously zero, to non-zero with a leading $1$ in the $(t+1)$\textsuperscript{th} column.
	\item Remove anti-diagonals until we reach the one that goes through index $j,t+2$.
\end{enumerate}
Observe that each matrix $U_t$ has rank $t$; thus, $U_s$ has rank $s$. Finally, applying parallel $CX$ gates to remove the remaining anti-diagonals of $U_s$, we obtain the identity matrix $I$.
In this algorithm, each anti-diagonal is removed exactly once (each removal giving a depth-$1$ $CX$ circuit), and each $CCX$ gate adds $1$ to the rank of $U$, so we need $s-r\leq s$ of them. More precisely, we need $s-\rank(M_\chi)$ Toffolis, and since $\rank(M_\chi)$ is at least $\log_2(s)$, the resulting Toffoli count is at most $s-\log_2(s)$.

Since all the gates applied here are involutive, it suffices to apply them in the other direction to build circuit $U_\chi$. Overall it has depth $\mathcal O(s+n)$, $\mathcal O(s(s+n))$ gate count, and $s-r\leq\mathcal O(s)$ Toffoli-count. It uses no non-Clifford+T gate. The algorithm is given in \Cref{alg:upper-elim-comp}.

Now, each $CCX$ gate is used to turn an all-$0$ row to a non-zero row. When applied in the other direction, the target row becomes null, meaning that the corresponding qubit is set to $0$. The $CCX$ gates can hence all be replaced by the following circuit:
\[
\begin{tikzpicture}
	\begin{pgfonlayer}{nodelayer}
		\node [style=none] (0) at (-1.5, -0.5) {};
		\node [style=none] (1) at (-1.5, 0.5) {};
		\node [style=none] (2) at (1.25, 0.5) {};
		\node [style=dot] (3) at (0.25, 0.5) {};
		\node [style=none] (4) at (-1.5, 0) {};
		\node [style=none] (5) at (1.25, 0) {};
		\node [style=dot] (6) at (0.25, 0) {};
		\node [style=none] (7) at (-0.675, -0.5) {};
		\node [style=none] (8) at (-0.675, -0.325) {};
		\node [style=none] (9) at (-0.675, -0.675) {};
		\node [style=none] (10) at (-0.075, -0.325) {};
		\node [style=none] (11) at (-0.075, -0.675) {};
		\node [style=none] (12) at (-0.6, -0.6) {};
		\node [style=none] (13) at (-0.15, -0.6) {};
		\node [style=none] (14) at (-0.375, -0.6) {};
		\node [style=none] (15) at (-0.25, -0.4) {};
		\node [style=none] (16) at (-0.075, -0.475) {};
		\node [style=none] (17) at (0.225, -0.475) {};
		\node [style=none] (18) at (-0.075, -0.525) {};
		\node [style=none] (19) at (0.275, -0.525) {};
		\node [style=none] (20) at (0.225, 0.05) {};
		\node [style=none] (21) at (0.275, -0.025) {};
		\node [style=box] (22) at (-1.125, -0.5) {$H$};
		\node [style=none] (23) at (0.75, -0.5) {$\ket0$};
		\node [style=none] (24) at (1, -0.5) {};
		\node [style=none] (25) at (1.25, -0.5) {};
	\end{pgfonlayer}
	\begin{pgfonlayer}{edgelayer}
		\draw (1.center) to (2.center);
		\draw (4.center) to (5.center);
		\draw (9.center) to (11.center);
		\draw (11.center) to (10.center);
		\draw (10.center) to (8.center);
		\draw (8.center) to (9.center);
		\draw [bend left=60] (12.center) to (13.center);
		\draw [style={arrows={->[]}}] (14.center) to (15.center);
		\draw (16.center) to (17.center);
		\draw (18.center) to (19.center);
		\draw (0.center) to (7.center);
		\draw (3) to (6);
		\draw (19.center) to (21.center);
		\draw (20.center) to (17.center);
		\draw (25.center) to (24.center);
	\end{pgfonlayer}
\end{tikzpicture}
\]

\begin{example}
	\label{ex:permut}
	Suppose $\chi = \{100001, 010010, 001101, 110011, 101100, 011111, 111110\}$. We then have:
	\[M_\chi = \begin{pmatrix}
		1&0&0&1&1&0&1\\
		0&1&0&1&0&1&1\\
		0&0&1&0&1&1&1\\
		0&0&1&0&1&1&1\\
		0&1&0&1&0&1&1\\
		1&0&1&1&0&1&0\\
		0&0&0&0&0&0&0
	\end{pmatrix}
	=
	\begin{pmatrix}
		1& & & & & & \\
		 &1& & & & & \\
		 & &1& & & & \\
		 & &1&1& & & \\
		 &1& & &1& & \\
		1& &1& & &1& \\
		 & & & & & &1
	\end{pmatrix}
	\begin{pmatrix}
		1& & &1&1& &1\\
		 &1& &1& &1&1\\
		 & &1& &1&1&1\\
		 & & & & & & \\
		 & & & & & & \\
		 & & & & & & \\
		 & & & & & & 
	\end{pmatrix}\]
	The detail of the execution of the algorithm on $L$ and $U$ can be found at \Cref{sec:details-example-permut}. 
	In this example, no permutation arises from either the PLUQ decomposition or from the circuit construction from $U$. Aggregating the gates from the end to the start, we get one of the two following circuits, depending whether we allow intermediary measurements or not:
	\[\scalebox{0.65}{
\begin{tikzpicture}
	\begin{pgfonlayer}{nodelayer}
		\node [style=none] (0) at (-3, 1) {};
		\node [style=none] (1) at (-3, 1.5) {};
		\node [style=none] (2) at (4.75, 1.5) {};
		\node [style=none] (3) at (-2.375, 0) {$\oplus$};
		\node [style=none] (4) at (4.75, 1) {};
		\node [style=dot] (5) at (-2.375, -1.5) {};
		\node [style=none] (6) at (-3, 0) {};
		\node [style=none] (7) at (-3, 0.5) {};
		\node [style=none] (8) at (4.75, 0.5) {};
		\node [style=none] (9) at (4.75, 0) {};
		\node [style=none] (10) at (-3, -1) {};
		\node [style=none] (11) at (-3, -0.5) {};
		\node [style=none] (12) at (4.75, -0.5) {};
		\node [style=none] (13) at (4.75, -1) {};
		\node [style=none] (15) at (-3, -1.5) {};
		\node [style=none] (16) at (4, -1.5) {};
		\node [style=none] (17) at (4.25, -1) {$\oplus$};
		\node [style=dot] (18) at (4.25, 0.5) {};
		\node [style=none] (19) at (3.75, 0) {$\oplus$};
		\node [style=dot] (20) at (3.75, 0.5) {};
		\node [style=none] (21) at (3.75, -0.5) {$\oplus$};
		\node [style=dot] (22) at (3.75, 1) {};
		\node [style=none] (23) at (3.75, -1) {$\oplus$};
		\node [style=dot] (24) at (3.75, 1.5) {};
		\node [style=none] (25) at (2.75, 0) {$\oplus$};
		\node [style=dot] (26) at (2.75, 1.5) {};
		\node [style=dot] (27) at (2.75, 1) {};
		\node [style=none] (28) at (2.25, 1.5) {$\oplus$};
		\node [style=dot] (29) at (2.25, 0) {};
		\node [style=none] (30) at (1.75, -0.5) {$\oplus$};
		\node [style=dot] (31) at (1.75, 1.5) {};
		\node [style=dot] (32) at (1.75, 0.5) {};
		\node [style=none] (33) at (1.25, 1) {$\oplus$};
		\node [style=dot] (34) at (1.25, 0) {};
		\node [style=none] (35) at (1.25, 1.5) {$\oplus$};
		\node [style=dot] (36) at (1.25, -0.5) {};
		\node [style=none] (37) at (0.5, -1) {$\oplus$};
		\node [style=dot] (38) at (0.5, 1) {};
		\node [style=dot] (39) at (0.5, 0.5) {};
		\node [style=none] (40) at (0, 0.5) {$\oplus$};
		\node [style=dot] (41) at (0, -0.5) {};
		\node [style=none] (42) at (0, 1) {$\oplus$};
		\node [style=dot] (43) at (0, -1) {};
		\node [style=none] (44) at (-0.75, 0.5) {$\oplus$};
		\node [style=dot] (45) at (-0.75, -1) {};
		\node [style=none] (46) at (-1.25, -1.5) {$\oplus$};
		\node [style=dot] (47) at (-1.25, 0.5) {};
		\node [style=dot] (48) at (-1.25, 0) {};
		\node [style=none] (49) at (-1.75, 0.5) {$\oplus$};
		\node [style=dot] (50) at (-1.75, -1.5) {};
		\node [style=none] (51) at (3.25, -1.75) {};
		\node [style=none] (52) at (3.25, 1.75) {};
		\node [style=none] (53) at (0, 1.75) {$U$};
		\node [style=none] (54) at (4, 1.75) {$L$};
		\node [style=none] (55) at (4.5, -1.5) {$\ket0$};
	\end{pgfonlayer}
	\begin{pgfonlayer}{edgelayer}
		\draw (1.center) to (2.center);
		\draw (4.center) to (0.center);
		\draw (3.center) to (5);
		\draw (7.center) to (8.center);
		\draw (9.center) to (6.center);
		\draw (11.center) to (12.center);
		\draw (13.center) to (10.center);
		\draw (15.center) to (16.center);
		\draw (17.center) to (18);
		\draw (19.center) to (20);
		\draw [bend left] (21.center) to (22);
		\draw [bend left] (23.center) to (24);
		\draw (25.center) to (26);
		\draw (28.center) to (29);
		\draw (30.center) to (31);
		\draw (33.center) to (34);
		\draw [bend right] (35.center) to (36);
		\draw (37.center) to (38);
		\draw (40.center) to (41);
		\draw [bend right] (42.center) to (43);
		\draw (44.center) to (45);
		\draw (46.center) to (47);
		\draw (49.center) to (50);
		\draw [style=dashed, draw=gray] (52.center) to (51.center);
	\end{pgfonlayer}
\end{tikzpicture}
}\qquad\scalebox{0.65}{
\begin{tikzpicture}
	\begin{pgfonlayer}{nodelayer}
		\node [style=none] (0) at (-8.25, 1) {};
		\node [style=none] (1) at (-8.25, 1.5) {};
		\node [style=none] (2) at (3.5, 1.5) {};
		\node [style=none] (3) at (-7.625, 0) {$\oplus$};
		\node [style=none] (4) at (3.5, 1) {};
		\node [style=dot] (5) at (-7.625, -1.5) {};
		\node [style=none] (7) at (-8.25, 0.5) {};
		\node [style=none] (8) at (3.5, 0.5) {};
		\node [style=none] (16) at (3, -1) {$\oplus$};
		\node [style=dot] (17) at (3, 0.5) {};
		\node [style=none] (18) at (2.5, 0) {$\oplus$};
		\node [style=dot] (19) at (2.5, 0.5) {};
		\node [style=none] (20) at (2.5, -0.5) {$\oplus$};
		\node [style=dot] (21) at (2.5, 1) {};
		\node [style=none] (22) at (2.5, -1) {$\oplus$};
		\node [style=dot] (23) at (2.5, 1.5) {};
		\node [style=none] (27) at (-0.5, 1.5) {$\oplus$};
		\node [style=dot] (28) at (-0.5, 0) {};
		\node [style=none] (32) at (-2.375, 1) {$\oplus$};
		\node [style=dot] (33) at (-2.375, 0) {};
		\node [style=none] (34) at (-2.375, 1.5) {$\oplus$};
		\node [style=dot] (35) at (-2.375, -0.5) {};
		\node [style=none] (39) at (-4.375, 0.5) {$\oplus$};
		\node [style=dot] (40) at (-4.375, -0.5) {};
		\node [style=none] (41) at (-4.375, 1) {$\oplus$};
		\node [style=dot] (42) at (-4.375, -1) {};
		\node [style=none] (43) at (-5.125, 0.5) {$\oplus$};
		\node [style=dot] (44) at (-5.125, -1) {};
		\node [style=none] (48) at (-7, 0.5) {$\oplus$};
		\node [style=dot] (49) at (-7, -1.5) {};
		\node [style=none] (50) at (2, -1.5) {};
		\node [style=none] (51) at (2, 1.75) {};
		\node [style=none] (52) at (-4.375, 1.75) {$U$};
		\node [style=none] (53) at (2.75, 1.75) {$L$};
		\node [style=none] (54) at (1.75, 0) {};
		\node [style=none] (55) at (3.5, 0) {};
		\node [style=none] (62) at (-8.25, -1) {};
		\node [style=dot] (65) at (-3, 1) {};
		\node [style=dot] (68) at (-3, 0.5) {};
		\node [style=none] (69) at (-3.3, -1) {};
		\node [style=none] (70) at (-3.3, -0.825) {};
		\node [style=none] (71) at (-3.3, -1.175) {};
		\node [style=none] (72) at (-2.7, -0.825) {};
		\node [style=none] (73) at (-2.7, -1.175) {};
		\node [style=none] (74) at (-3.225, -1.1) {};
		\node [style=none] (75) at (-2.775, -1.1) {};
		\node [style=none] (76) at (-3, -1.1) {};
		\node [style=none] (77) at (-2.875, -0.9) {};
		\node [style=none] (79) at (-3.025, -0.825) {};
		\node [style=none] (81) at (-2.975, -0.825) {};
		\node [style=none] (82) at (-3.025, 0.55) {};
		\node [style=none] (83) at (-2.975, 0.475) {};
		\node [style=box] (84) at (-3.75, -1) {$H$};
		\node [style=none] (85) at (-8.25, -1.5) {};
		\node [style=dot] (86) at (-5.625, 0.5) {};
		\node [style=dot] (87) at (-5.625, 0) {};
		\node [style=none] (88) at (-5.925, -1.5) {};
		\node [style=none] (89) at (-5.925, -1.325) {};
		\node [style=none] (90) at (-5.925, -1.675) {};
		\node [style=none] (91) at (-5.325, -1.325) {};
		\node [style=none] (92) at (-5.325, -1.675) {};
		\node [style=none] (93) at (-5.85, -1.6) {};
		\node [style=none] (94) at (-5.4, -1.6) {};
		\node [style=none] (95) at (-5.625, -1.6) {};
		\node [style=none] (96) at (-5.5, -1.4) {};
		\node [style=none] (98) at (-5.65, -1.325) {};
		\node [style=none] (100) at (-5.6, -1.325) {};
		\node [style=none] (101) at (-5.65, 0.05) {};
		\node [style=none] (102) at (-5.6, -0.025) {};
		\node [style=box] (103) at (-6.375, -1.5) {$H$};
		\node [style=none] (104) at (-8.25, -0.5) {};
		\node [style=dot] (105) at (-1, 1.5) {};
		\node [style=dot] (106) at (-1, 0.5) {};
		\node [style=none] (107) at (-1.3, -0.5) {};
		\node [style=none] (108) at (-1.3, -0.325) {};
		\node [style=none] (109) at (-1.3, -0.675) {};
		\node [style=none] (110) at (-0.7, -0.325) {};
		\node [style=none] (111) at (-0.7, -0.675) {};
		\node [style=none] (112) at (-1.225, -0.6) {};
		\node [style=none] (113) at (-0.775, -0.6) {};
		\node [style=none] (114) at (-1, -0.6) {};
		\node [style=none] (115) at (-0.875, -0.4) {};
		\node [style=none] (117) at (-1.025, -0.325) {};
		\node [style=none] (119) at (-0.975, -0.325) {};
		\node [style=none] (120) at (-1.025, 0.55) {};
		\node [style=none] (121) at (-0.975, 0.475) {};
		\node [style=box] (122) at (-1.75, -0.5) {$H$};
		\node [style=none] (123) at (-8.25, 0) {};
		\node [style=dot] (124) at (0.875, 1.5) {};
		\node [style=dot] (125) at (0.875, 1) {};
		\node [style=none] (126) at (0.575, 0) {};
		\node [style=none] (127) at (0.575, 0.175) {};
		\node [style=none] (128) at (0.575, -0.175) {};
		\node [style=none] (129) at (1.175, 0.175) {};
		\node [style=none] (130) at (1.175, -0.175) {};
		\node [style=none] (131) at (0.65, -0.1) {};
		\node [style=none] (132) at (1.1, -0.1) {};
		\node [style=none] (133) at (0.875, -0.1) {};
		\node [style=none] (134) at (1, 0.1) {};
		\node [style=none] (136) at (0.85, 0.175) {};
		\node [style=none] (138) at (0.9, 0.175) {};
		\node [style=none] (139) at (0.85, 1.05) {};
		\node [style=none] (140) at (0.9, 0.975) {};
		\node [style=box] (141) at (0.125, 0) {$H$};
		\node [style=none] (142) at (1.5, 0) {$\ket0$};
		\node [style=none] (143) at (1.75, -0.5) {};
		\node [style=none] (144) at (3.5, -0.5) {};
		\node [style=none] (145) at (1.5, -0.5) {$\ket0$};
		\node [style=none] (146) at (1.75, -1) {};
		\node [style=none] (147) at (3.5, -1) {};
		\node [style=none] (148) at (1.5, -1) {$\ket0$};
	\end{pgfonlayer}
	\begin{pgfonlayer}{edgelayer}
		\draw (1.center) to (2.center);
		\draw (4.center) to (0.center);
		\draw (3.center) to (5);
		\draw (7.center) to (8.center);
		\draw (16.center) to (17);
		\draw (18.center) to (19);
		\draw [bend left] (20.center) to (21);
		\draw [bend left] (22.center) to (23);
		\draw (27.center) to (28);
		\draw (32.center) to (33);
		\draw [bend right] (34.center) to (35);
		\draw (39.center) to (40);
		\draw [bend right] (41.center) to (42);
		\draw (43.center) to (44);
		\draw (48.center) to (49);
		\draw [style=dashed, draw=gray] (51.center) to (50.center);
		\draw (55.center) to (54.center);
		\draw (71.center) to (73.center);
		\draw (73.center) to (72.center);
		\draw (72.center) to (70.center);
		\draw (70.center) to (71.center);
		\draw [bend left=60] (74.center) to (75.center);
		\draw [style={arrows={->[]}}] (76.center) to (77.center);
		\draw (62.center) to (69.center);
		\draw (65) to (68);
		\draw (81.center) to (83.center);
		\draw (82.center) to (79.center);
		\draw (90.center) to (92.center);
		\draw (92.center) to (91.center);
		\draw (91.center) to (89.center);
		\draw (89.center) to (90.center);
		\draw [bend left=60] (93.center) to (94.center);
		\draw [style={arrows={->[]}}] (95.center) to (96.center);
		\draw (85.center) to (88.center);
		\draw (86) to (87);
		\draw (100.center) to (102.center);
		\draw (101.center) to (98.center);
		\draw (109.center) to (111.center);
		\draw (111.center) to (110.center);
		\draw (110.center) to (108.center);
		\draw (108.center) to (109.center);
		\draw [bend left=60] (112.center) to (113.center);
		\draw [style={arrows={->[]}}] (114.center) to (115.center);
		\draw (104.center) to (107.center);
		\draw (105) to (106);
		\draw (119.center) to (121.center);
		\draw (120.center) to (117.center);
		\draw (128.center) to (130.center);
		\draw (130.center) to (129.center);
		\draw (129.center) to (127.center);
		\draw (127.center) to (128.center);
		\draw [bend left=60] (131.center) to (132.center);
		\draw [style={arrows={->[]}}] (133.center) to (134.center);
		\draw (123.center) to (126.center);
		\draw (124) to (125);
		\draw (138.center) to (140.center);
		\draw (139.center) to (136.center);
		\draw (144.center) to (143.center);
		\draw (147.center) to (146.center);
	\end{pgfonlayer}
\end{tikzpicture}
}\]
\end{example}

\section{W-State Preparation}
\label{sec:w-state}

We now aim to use the previous construction to prepare sparse states. Since it maps $\ket{e_i}$ to arbitrary target classical outputs, it becomes natural to build states of the form $\sum_i\alpha_i\ket{e_i}$. These are called weighted $W$ states. The $W$ state emerged as the canonical representative of the only two non-trivial classes of entanglement on 3 qubits \cite{Dur2000Slocc}. Generalised to $n$-qubits, it is expressed as:
\begin{align*}
	\ket{W_n} 
	&:= 
	\frac{1}{\sqrt{n}}
	\Big(
	\ket{10...0}+\ket{01...0}+...+\ket{00...1}
	\Big) = \frac{1}{\sqrt{n}}\sum_i\ket{e_i}
\end{align*}
And further generalised to arbitrary weights $\vec\alpha := (\alpha_1,\ldots,\alpha_n) \in \mathbb{C}^n$:
\begin{align*}
	\ket{W(\vec\alpha)}
	&:=\frac1{\sqrt{\sum_i |\alpha_i|^2}}\Big(
	\alpha_1\ket{10...0} +
	\alpha_2\ket{01...0} +
	\ldots +
	\alpha_n\ket{00...1}\Big) = \frac1{\sqrt{\sum_i |\alpha_i|^2}}\sum_i\alpha_i\ket{e_i}
\end{align*}
Of course, the first is a special case of the second: $\ket{W_n}=\ket{W(1,...,1)}$. 
We give in the following two different constructions for this family of states. The first uses the following inductive redefinition:
\begin{align*}
	&\ket{W(\rho e^{i\phi})} = e^{i\phi}\ket1 \\
	&\ket{W(\vec\alpha\oplus\vec\beta)} =
	\frac1{\sqrt{w(\vec\alpha)+w(\vec\beta)}}
	\left(
	\sqrt{w(\vec\alpha)}\ket{W(\vec\alpha)}\ket0^{\otimes |\vec\beta|} +
	\sqrt{w(\vec\beta)}\ket0^{\otimes|\vec\alpha|}\ket{W(\vec\beta)}
	\right)
\end{align*}
where $w(\vec\alpha):=\sum_i |\alpha_i|^2$ and $w(\vec\beta):=\sum_i |\beta_i|^2$, $\vec\alpha\oplus\vec\beta$ is the concatenation of $\vec\alpha$ and $\vec\beta$, and $|\vec\alpha|$ and $|\vec\beta|$ are the cardinalities of $\vec\alpha$ and $\vec\beta$. This will yield an exact, ancilla-free construction. The second method will start from a full compact state preparation that will be modified into a $W$-state. It will be more amenable to Clifford+T approximations, but will make use of ancillae.

\subsection{Tree-based construction}

We start by explaining the exact construction that uses the inductive definition above. This construction is not new \cite{Johri2021Nearest}, but we dive a bit further in the analysis of resources its requires.

The construction will be explained in all generality for arbitrary weights, and the result will of course be specialisable to the uniform case. The structure of the result will be based on a tree whose leaves bear the weights.

We define a tree $t$ as either a leaf $\bot_i$ with an amplitude $w_i\in\mathbb R^+$ and a phase $\phi_i\in\mathbb R^+$ for some index $i$; or a pair of trees $(t_1,t_2)$. The height $h(t)$, leaf count $\ell(t)$ and weight $w(t)$ of a tree $t$ are easily defined inductively as:
\begin{align*}
	\begin{array}{rl@{\qquad\quad}l}
		h(\bot_i) &= 0
		& \ell(\bot_i) = 1\quad
		 w(\bot_i) = w_i\\[0.5em]
		h((t_1,t_2)) &= \max(h(t_1),h(t_2))+1
		& \operatorname{op}((t_1,t_2)) = \operatorname{op}(t_1)+\operatorname{op}(t_2)
		~\text{ for }~\operatorname{op}\in\{\ell,w\}
	\end{array}
\end{align*}

Given a tree $t$, we build a corresponding circuit $C(t)$, with $1$ input qubit and $\ell(t)$ output qubits, inductively as follows:
\begin{gather*}
	\bot_i \mapsto 
\begin{tikzpicture}
	\begin{pgfonlayer}{nodelayer}
		\node [style=none] (0) at (-0.75, 0) {};
		\node [style=box] (1) at (0, 0) {$P(\phi_i)$};
		\node [style=none] (2) at (0.75, 0) {};
	\end{pgfonlayer}
	\begin{pgfonlayer}{edgelayer}
		\draw (2.center) to (0.center);
	\end{pgfonlayer}
\end{tikzpicture}

	\qquad\qquad
	(t_1,t_2) \mapsto 
\begin{tikzpicture}
	\begin{pgfonlayer}{nodelayer}
		\node [style=none] (0) at (0, -0.125) {};
		\node [style=none] (1) at (1.5, -0.125) {};
		\node [style=none] (2) at (0, -0.875) {};
		\node [style=none] (3) at (1.5, -0.875) {};
		\node [style=none] (4) at (0.75, -0.5) {$C(t_2)$};
		\node [style=none] (5) at (0, -0.5) {};
		\node [style=none] (6) at (-0.25, -0.5) {};
		\node [style=none] (7) at (1.75, -0.75) {};
		\node [style=none] (8) at (1.5, -0.75) {};
		\node [style=none] (9) at (1.75, -0.25) {};
		\node [style=none] (10) at (1.5, -0.25) {};
		\node [style=none] (11) at (0, 0.875) {};
		\node [style=none] (12) at (1.5, 0.875) {};
		\node [style=none] (13) at (0, 0.125) {};
		\node [style=none] (14) at (1.5, 0.125) {};
		\node [style=none] (15) at (0.75, 0.5) {$C(t_1)$};
		\node [style=none] (16) at (0, 0.5) {};
		\node [style=none] (17) at (-0.25, 0.5) {};
		\node [style=none] (18) at (1.75, 0.25) {};
		\node [style=none] (19) at (1.5, 0.25) {};
		\node [style=none] (20) at (1.75, 0.75) {};
		\node [style=none] (21) at (1.5, 0.75) {};
		\node [style=none] (22) at (1.65, 0.6) {$\vdots$};
		\node [style=none] (23) at (1.65, -0.4) {$\vdots$};
		\node [style=none] (24) at (-1.5, 0.625) {};
		\node [style=none] (25) at (-0.25, 0.625) {};
		\node [style=none] (26) at (-1.5, -0.625) {};
		\node [style=none] (27) at (-0.25, -0.625) {};
		\node [style=none] (28) at (-0.875, 0) {$F_{\frac{w(t_2)}{w(t)}}$};
		\node [style=none] (29) at (-1.5, 0) {};
		\node [style=none] (30) at (-1.75, 0) {};
	\end{pgfonlayer}
	\begin{pgfonlayer}{edgelayer}
		\draw (3.center) to (1.center);
		\draw (1.center) to (0.center);
		\draw (0.center) to (2.center);
		\draw (2.center) to (3.center);
		\draw (6.center) to (5.center);
		\draw (8.center) to (7.center);
		\draw (10.center) to (9.center);
		\draw (14.center) to (12.center);
		\draw (12.center) to (11.center);
		\draw (11.center) to (13.center);
		\draw (13.center) to (14.center);
		\draw (17.center) to (16.center);
		\draw (19.center) to (18.center);
		\draw (21.center) to (20.center);
		\draw (27.center) to (25.center);
		\draw (25.center) to (24.center);
		\draw (24.center) to (26.center);
		\draw (26.center) to (27.center);
		\draw (30.center) to (29.center);
	\end{pgfonlayer}
\end{tikzpicture}
\\
	\text{where}
	\qquad
	
\begin{tikzpicture}
	\begin{pgfonlayer}{nodelayer}
		\node [style=none] (0) at (-1.25, 0.375) {};
		\node [style=none] (1) at (-0.75, 0.375) {};
		\node [style=none] (2) at (-1.25, -0.375) {};
		\node [style=none] (3) at (-0.75, -0.375) {};
		\node [style=none] (4) at (-1, 0) {$F_p$};
		\node [style=none] (5) at (-1.25, 0) {};
		\node [style=none] (6) at (-1.5, 0) {};
		\node [style=none] (7) at (-0.5, -0.25) {};
		\node [style=none] (8) at (-0.75, -0.25) {};
		\node [style=none] (9) at (-0.5, 0.25) {};
		\node [style=none] (10) at (-0.75, 0.25) {};
	\end{pgfonlayer}
	\begin{pgfonlayer}{edgelayer}
		\draw (3.center) to (1.center);
		\draw (1.center) to (0.center);
		\draw (0.center) to (2.center);
		\draw (2.center) to (3.center);
		\draw (6.center) to (5.center);
		\draw (8.center) to (7.center);
		\draw (10.center) to (9.center);
	\end{pgfonlayer}
\end{tikzpicture}
~~=~~
\begin{tikzpicture}
	\begin{pgfonlayer}{nodelayer}
		\node [style=none] (0) at (-1.25, 0.25) {};
		\node [style=none] (1) at (-1.75, -0.25) {};
		\node [style=none] (2) at (0.75, -0.25) {};
		\node [style=box] (3) at (-0.5, 0.25) {$R_Y(\theta)$};
		\node [style=none] (4) at (-1.5, 0.25) {$\ket0$};
		\node [style=none] (5) at (0.75, 0.25) {};
		\node [style=dot] (6) at (-0.5, -0.25) {};
		\node [style=dot] (7) at (0.25, 0.25) {};
		\node [style=none] (8) at (0.25, -0.25) {$\oplus$};
	\end{pgfonlayer}
	\begin{pgfonlayer}{edgelayer}
		\draw (1.center) to (2.center);
		\draw (5.center) to (0.center);
		\draw (3) to (6);
		\draw (7) to (8.center);
	\end{pgfonlayer}
\end{tikzpicture}

	\qquad
	\text{with}
	\qquad
	\theta := 2\arccos(\sqrt{p})
\end{gather*}
The base component of the construction is the circuit $F_p$, which was defined in \cite{Diker2022Deterministic,Johri2021Nearest} and generalized to qudits in \cite{Yeh2023Scaling}. It is a circuit parametrized by $p\in[0,1]$, where $p$ is chosen to be the ratio between the weight $w(t_2)$ of subtree $t_2$ and the overall weight $w(t)$ of tree $t$. The circuit implements a $1$-qubit input, $2$-qubit output isometry map, with the following matrix and Dirac notation representations:
\begin{gather*}
	
\begin{tikzpicture}
	\begin{pgfonlayer}{nodelayer}
		\node [style=none] (0) at (-1.25, 0.375) {};
		\node [style=none] (1) at (-0.75, 0.375) {};
		\node [style=none] (2) at (-1.25, -0.375) {};
		\node [style=none] (3) at (-0.75, -0.375) {};
		\node [style=none] (4) at (-1, 0) {$F_p$};
		\node [style=none] (5) at (-1.25, 0) {};
		\node [style=none] (6) at (-1.5, 0) {};
		\node [style=none] (7) at (-0.5, -0.25) {};
		\node [style=none] (8) at (-0.75, -0.25) {};
		\node [style=none] (9) at (-0.5, 0.25) {};
		\node [style=none] (10) at (-0.75, 0.25) {};
	\end{pgfonlayer}
	\begin{pgfonlayer}{edgelayer}
		\draw (3.center) to (1.center);
		\draw (1.center) to (0.center);
		\draw (0.center) to (2.center);
		\draw (2.center) to (3.center);
		\draw (6.center) to (5.center);
		\draw (8.center) to (7.center);
		\draw (10.center) to (9.center);
	\end{pgfonlayer}
\end{tikzpicture}

	~~=~~
	\begin{pmatrix}
		1&0\\0&\sqrt p\\0&\sqrt{1-p}\\0&0
	\end{pmatrix}
	~~=~~
	\ketbra{00}{0}+\sqrt{p}\ketbra{01}{1}+\sqrt{1-p}\ketbra{10}{1}
\end{gather*}

This circuit allows us to build the desired state, and its metrics depend on the choice of the tree:

\begin{proposition}
\label{prop:semantics-W-prep-weight}
Suppose $\vec\alpha := (\rho_ie^{i\phi_k})_{k \in \{1,\ldots,n\}}\in \mathbb{C}^n$ is a list of complex numbers in polar form, for some $n \in \mathbb{N}$, and $t$ is an $n$-leaves tree with weights $\vec{w} := (\rho_k^2)_{k \in \{1,\ldots,n\}}$ and phases $\vec\phi := (\phi_k)_{k \in \{1,\ldots,n\}}$. Let $h(t)$ be the height of the tree. Then:
\begin{itemize}
	\item The circuit $C(t)$ maps $\ket1$ to $\ket{W(\vec\alpha)}$.
	\item The depth of circuit $C(t)$ is $\mathcal O(h(t))$ and its gate count is $\mathcal O(n)$.
\end{itemize}
\end{proposition}

As a consequence, by choosing the appropriate tree:

\begin{corollary}
\label{thm:Wn-weight}
For any $n \in \mathbb{N}$ and $\vec\alpha \in \mathbb{C}^n$, the $n$-qubit weighted W-state $\ket{W(\vec\alpha)}$ can be prepared by a circuit with $\mathcal{O}(n)$ size, $\mathcal{O}(\log_2(n))$ depth, and $\mathcal{O}(n)$ non-Clifford gates.
\end{corollary}

All $R_Y$ gates used here use arbitrary angles, and are hence outside the Clifford+T gate set. There exist schemes to compile them down to Clifford+T (e.g.~\cite{Selinger2015Efficient}), but before we do so, it is worth checking when exactly is $F_p$ in this gate set.

\begin{proposition}
	\label{prop:Fp-resources}
	The circuit $F_p$ can be implemented in:
	\begin{itemize}
		\item the Clifford gate-set iff $p\in\{0,1\}$
		\item the Clifford+T gate set iff $p\in\{0,\frac12,1\}$
	\end{itemize}
\end{proposition}

When $p\in\{0,1\}$, $F_p$ specialises to $I\otimes\ket0$ and $\ket0\otimes I$ respectively. When $p=1/2$, the $R_Y$ gate can be replaced by the CH gate, which indeed is in Clifford+T. The result above shows that in any other case, $F_p$ requires non-Clifford+T gates to be built exactly. The case $p=1/2$ appears frequently when dealing with uniform $W$ states. More precisely:

\begin{corollary}
	\label{thm:Wn}
	For any $n \in \mathbb{N}$, an $n$-qubit uniform W-state $\ket{W_n}$ can be prepared by a circuit with $\mathcal{O}(n)$ circuit size, $\mathcal{O}(\log_2(n))$ circuit depth, $\mathcal{O}(n)$ T gates, and $\mathcal{O}(\log_2(n))$ non-Clifford+T gates.
\end{corollary}

More generally, an occurrence of $F_p$ with $p=1/2$ occurs when considering a tree $(t_1,t_2)$ such that $w(t_1)=w(t_2)$. Let's call such a tree, a \emph{perfectly balanced} tree. It is then possible to choose the tree to build $\ket{W(\vec\alpha)}$ in a way that maximises the number of perfectly balanced subtrees, at the expense of depth (see for example \Cref{fig:example-different-trees} to build $\ket{W(\vec\alpha)}$ where $\vec\alpha = (\sqrt3,\sqrt2,\sqrt2,1,1,1,1)$.\\
\begin{figure}[!ht]
	$a.~~
\begin{tikzpicture}
	\begin{pgfonlayer}{nodelayer}
		\node [style=dot] (0) at (-1.75, -0.5) {};
		\node [style=dot] (1) at (-1.25, -0.5) {};
		\node [style=dot] (2) at (-0.75, -0.5) {};
		\node [style=dot] (3) at (-0.25, -0.5) {};
		\node [style=dot] (4) at (0.25, -0.5) {};
		\node [style=dot] (5) at (0.75, -0.5) {};
		\node [style=dot] (6) at (1.5, 0) {};
		\node [style=dot] (7) at (-1.5, 0) {};
		\node [style=dot] (8) at (-0.5, 0) {};
		\node [style=dot] (9) at (0.5, 0) {};
		\node [style=dot] (10) at (-1, 0.5) {};
		\node [style=dot] (11) at (1, 0.5) {};
		\node [style=dot] (12) at (0, 1) {};
		\node [style=none] (13) at (-1.75, -0.75) {$3$};
		\node [style=none] (14) at (-1.25, -0.75) {$2$};
		\node [style=none] (15) at (-0.75, -0.75) {$2$};
		\node [style=none] (16) at (-0.25, -0.75) {$1$};
		\node [style=none] (17) at (0.25, -0.75) {$1$};
		\node [style=none] (18) at (0.75, -0.75) {$1$};
		\node [style=none] (19) at (1.5, -0.25) {$1$};
		\node [style=highlight] (20) at (-1.5, 0) {};
		\node [style=highlight] (21) at (-0.5, 0) {};
		\node [style=highlight] (22) at (1, 0.5) {};
		\node [style=highlight] (23) at (-1, 0.5) {};
		\node [style=highlight] (24) at (0, 1) {};
		\node [style=none, text=gray] (25) at (-1.75, 0.25) {$\frac25$};
		\node [style=none, text=gray] (26) at (-0.25, 0.25) {$\frac13$};
		\node [style=none, text=gray] (27) at (0.25, 0.25) {$\frac12$};
		\node [style=none, text=gray] (28) at (1.25, 0.75) {$\frac13$};
		\node [style=none, text=gray] (29) at (-1.25, 0.75) {$\frac38$};
		\node [style=none, text=gray] (30) at (0.25, 1.25) {$\frac3{11}$};
	\end{pgfonlayer}
	\begin{pgfonlayer}{edgelayer}
		\draw (12) to (10);
		\draw (10) to (7);
		\draw (7) to (0);
		\draw (7) to (1);
		\draw (2) to (8);
		\draw (8) to (10);
		\draw (8) to (3);
		\draw (4) to (9);
		\draw (9) to (5);
		\draw (9) to (11);
		\draw (11) to (6);
		\draw (11) to (12);
	\end{pgfonlayer}
\end{tikzpicture}
\hfill b.~~
\begin{tikzpicture}
	\begin{pgfonlayer}{nodelayer}
		\node [style=dot] (0) at (-2.25, -0.5) {};
		\node [style=dot] (1) at (-1.75, -0.5) {};
		\node [style=dot] (2) at (-1.25, -0.5) {};
		\node [style=dot] (3) at (-0.75, -0.5) {};
		\node [style=dot] (4) at (-0.25, -0.5) {};
		\node [style=dot] (5) at (0.25, -0.5) {};
		\node [style=dot] (6) at (1, 0) {};
		\node [style=dot] (7) at (-2, 0) {};
		\node [style=dot] (8) at (-1, 0) {};
		\node [style=dot] (9) at (0, 0) {};
		\node [style=dot] (10) at (-1.5, 0.5) {};
		\node [style=dot] (11) at (0.5, 0.5) {};
		\node [style=dot] (12) at (-0.5, 1) {};
		\node [style=none] (13) at (1, -0.25) {$3$};
		\node [style=none] (14) at (0.25, -0.75) {$2$};
		\node [style=none] (15) at (-0.25, -0.75) {$2$};
		\node [style=none] (16) at (-1.25, -0.75) {$1$};
		\node [style=none] (17) at (-0.75, -0.75) {$1$};
		\node [style=none] (18) at (-1.75, -0.75) {$1$};
		\node [style=none] (19) at (-2.25, -0.75) {$1$};
		\node [style=highlight] (20) at (0.5, 0.5) {};
		\node [style=highlight] (21) at (-0.5, 1) {};
		\node [style=none, text=gray] (22) at (-2.25, 0.25) {$\frac12$};
		\node [style=none, text=gray] (23) at (-0.75, 0.25) {$\frac12$};
		\node [style=none, text=gray] (24) at (-0.25, 0.25) {$\frac12$};
		\node [style=none, text=gray] (25) at (-1.75, 0.75) {$\frac12$};
		\node [style=none, text=gray] (26) at (0.75, 0.75) {$\frac37$};
		\node [style=none, text=gray] (27) at (-0.25, 1.25) {$\frac7{11}$};
	\end{pgfonlayer}
	\begin{pgfonlayer}{edgelayer}
		\draw (12) to (10);
		\draw (10) to (7);
		\draw (7) to (0);
		\draw (7) to (1);
		\draw (2) to (8);
		\draw (8) to (10);
		\draw (8) to (3);
		\draw (4) to (9);
		\draw (9) to (5);
		\draw (9) to (11);
		\draw (11) to (6);
		\draw (11) to (12);
	\end{pgfonlayer}
\end{tikzpicture}
\hfill c.~~
\begin{tikzpicture}
	\begin{pgfonlayer}{nodelayer}
		\node [style=dot] (0) at (-1.25, -0.75) {};
		\node [style=dot] (1) at (-0.75, -0.75) {};
		\node [style=dot] (2) at (-0.25, -0.75) {};
		\node [style=dot] (3) at (0.25, -0.75) {};
		\node [style=dot] (4) at (0.75, -0.25) {};
		\node [style=dot] (5) at (1.25, -0.25) {};
		\node [style=dot] (7) at (-1, -0.25) {};
		\node [style=dot] (8) at (0, -0.25) {};
		\node [style=dot] (9) at (1, 0.25) {};
		\node [style=dot] (10) at (-0.5, 0.25) {};
		\node [style=dot] (12) at (0.25, 0.75) {};
		\node [style=none] (14) at (1.25, -0.5) {$2$};
		\node [style=none] (15) at (0.75, -0.5) {$2$};
		\node [style=none] (16) at (-0.25, -1) {$1$};
		\node [style=none] (17) at (0.25, -1) {$1$};
		\node [style=none] (18) at (-0.75, -1) {$1$};
		\node [style=none] (19) at (-1.25, -1) {$1$};
		\node [style=dot] (20) at (1.75, 0.75) {};
		\node [style=none] (21) at (1.75, 0.5) {$3$};
		\node [style=dot] (22) at (1, 1) {};
		\node [style=highlight] (23) at (1, 1) {};
		\node [style=none, text=gray] (24) at (-1.25, 0) {$\frac12$};
		\node [style=none, text=gray] (25) at (0.25, 0) {$\frac12$};
		\node [style=none, text=gray] (26) at (1.25, 0.5) {$\frac12$};
		\node [style=none, text=gray] (27) at (-0.75, 0.5) {$\frac12$};
		\node [style=none, text=gray] (28) at (0, 1) {$\frac12$};
		\node [style=none, text=gray] (29) at (1.25, 1.25) {$\frac3{11}$};
	\end{pgfonlayer}
	\begin{pgfonlayer}{edgelayer}
		\draw (12) to (10);
		\draw (10) to (7);
		\draw (7) to (0);
		\draw (7) to (1);
		\draw (2) to (8);
		\draw (8) to (10);
		\draw (8) to (3);
		\draw (4) to (9);
		\draw (9) to (5);
		\draw (12) to (9);
		\draw (22) to (12);
		\draw (22) to (20);
	\end{pgfonlayer}
\end{tikzpicture}
$
	\caption{Weights $(3,2,2,1,1,1,1)$ arranged: a. in descending order in a complete tree, b. in ascending order in a complete tree, c. in an arbitrary tree.}
	\label{fig:example-different-trees}
\end{figure}
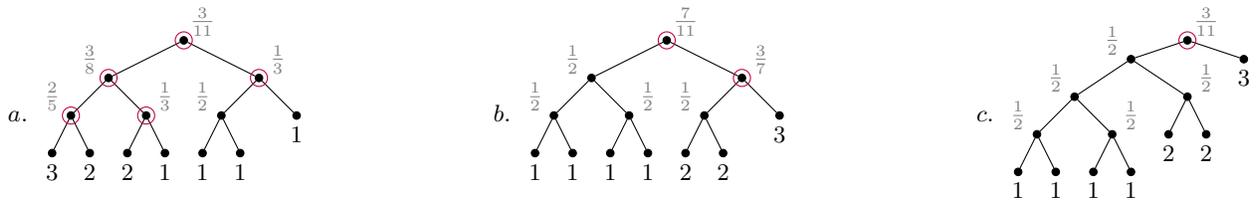
This optimisation problem is perfectly classical. We however conjecture that finding an optimal solution is \textbf{NP-hard}, due to its proximity with hard optimisation problems such as the partition problem. Its study is left open, as it goes beyond the current paper, which is interested in general bounds (and we can easily be convinced that in the worst cases, all $F_p$ will have to be outside the Clifford+T gate set).

In these worst cases, we can always use \cite{Dawson2006SolovayKitaev,Selinger2015Efficient} to compile down to Clifford+T up to some $\epsilon$ error:

\begin{theorem}
	\label{thm:Wn-weight-ClifT}
	For any $n \in \mathbb{N}$ and $\vec\alpha \in \mathbb{C}^n$, an $n$-qubit weighted W-state $\ket{W(\vec\alpha)}$ can be approximated by an ancilla-free Clifford+T circuit up to error $\epsilon$ with $\mathcal{O}(n\log_2(n/\epsilon))$ size, $\mathcal{O}(\log_2(n)\log_2(n/\epsilon))$ depth, and $\mathcal{O}(n\log_2(n/\epsilon))$ T gates.
\end{theorem}

\subsection{Full-state-based construction}

If T-count is the priority, it is possible to use another construction that reduces it, at the cost of a larger memory space. The idea is to perform a full-state preparation on $\lceil\log_2(n)\rceil$ qubits using \cite{Gosset2025Preparation} below\footnote{Depth is not provided in the reference, but a careful analysis of the construction provides the claimed complexity.}, and then apply the construction from \cref{sec:permut} to ``dispatch'' the values over the weight-1 basis states:
\begin{theorem}[\cite{Gosset2025Preparation}]
	\label{thm:Gosset}
	Let $\ket\psi$ be an $n$-qubit state, and $\epsilon>0$. The state $\ket\psi$ can be prepared up to error $\epsilon$ by a circuit with $\mathcal O\left(\sqrt{2^n\log_2(1/\epsilon)}+\log_2(1/\epsilon)\right)$ T-gates, ancillae, and depth. The T-count is optimal.
\end{theorem}

\begin{theorem}
	\label{thm:Wn-weight-from-compact}
	Let $\vec \alpha \in \mathbb C^n$ and $\epsilon >0$, then $\ket{W(\vec\alpha)}$ can be prepared up to error $\epsilon$ by a circuit using $\mathcal O(n+\log_2(1/\epsilon))$ T-gates and depth, and $\mathcal O\left(\sqrt{n\log_2(1/\epsilon)}+\log_2(1/\epsilon)\right) - n + \log_2(n)$ ancillas.
\end{theorem}

The dependency in $n$ in the $T$-count above is asymptotically optimal, as shown later in \cref{prop:lower-bound}.

\section{Sparse State Preparation}
\label{sec:sparse-state}

Our proposal for sparse state preparation is now a mere combination of the previous constructions, first, by creating a weighted W-state with appropriate weights, and then using the circuits for mapping the $\ket{e_i}$ states to the target states. First, assuming access to arbitrary-angled phase gates:

\begin{theorem}
	\label{thm:sparse}
	For any $s,n \in \mathbb{N}$, an $s$-sparse, $n$-qubit state $\ket{\psi}$ can be exactly prepared using a circuit with $\mathcal{O}(s(s+n))$ circuit size, $\mathcal{O}(s+n)$ circuit depth, $\mathcal{O}(s)$ non-Clifford gates, and $\max(0,s-n)$ ancillary qubits.
\end{theorem}

\begin{proof}
	By combining \Cref{thm:Wn-weight} and \Cref{thm:permut}.
\end{proof}

When compiling down to the Clifford+T gate set, one may trade depth and ancilla-count for T gates:


\begin{theorem}
\label{thm:sparse-state-prep-ClifT}
For any $s,n \in \mathbb{N}$, an $s$-sparse, $n$-qubit state $\ket{\psi}$ can be prepared up to error $\epsilon>0$ using either one of 3 circuits, respectively with the following metrics:\\
\begin{tabular}{@{}lll@{}}
	\toprule
	\textbf{T-count} & \textbf{Depth} & \textbf{Ancillae} \\
	\midrule
	$\mathcal{O}(s\log_2(s/\epsilon))$ & $\mathcal{O}(s+n+\log_2(s)\log_2(1/\epsilon))$ & $\max(0,s-n)$ \\
	$\mathcal O\left(s+\log_2(1/\epsilon)\right)$ & $\mathcal O\left(s+n+\log_2(1/\epsilon)\right)$ & $\mathcal O\left( s+\log_2(1/\epsilon) \right)$ \\
	\bottomrule
\end{tabular}
\end{theorem}

\begin{proof}[Proof of \Cref{thm:sparse-state-prep-ClifT}]
	The two constructions first build a weighted W-state, then use \Cref{thm:permut} to reach the target state. This second step adds no $T$-gate, and is of depth $\mathcal O(s+n)$. The first construction uses \Cref{thm:Wn-weight-ClifT}, while the second uses the result from \Cref{thm:Wn-weight-from-compact}. 
\end{proof}

Finally, we show a lower bound on the required number of T-gates required to build an arbitrary sparse state:
\begin{proposition}
	\label{prop:lower-bound}
	The T-count of arbitrary $s$-sparse $n$-qubit states is $\mathbf{\Omega}(\min(s,n))$.
\end{proposition}

As a result, we notice that the construction of \cref{thm:sparse-state-prep-ClifT} is asymptotically optimal for small sparsity, i.e.~when $s=\mathcal O(n)$.


\section{Conclusion}
\label{sec:discussion}

We presented in this paper several resource-efficient algorithms to prepare an arbitrary sparse state, either in the exact setting (with access to arbitrary-angled rotations), or in the fault-tolerant regime by compiling down to Clifford+T. Of the many metrics considered, our particular focus was the number of non-Clifford gates and the circuit depth. In the ideal setting, we showed that the number of non-Clifford gates is linear in the sparsity $s$, i.e., $\mathcal O(s)$, and that the circuit depth is linear in both sparsity $s$ and the number of qubits $n$, i.e., $\mathcal{O}(s+n)$. In the Clifford+T approximated setting, we provided several constructions whose depth and T-count vary depending on the number of ancillae used. The most T-wise optimised construction yields $\mathcal O\left(s+\log_2(1/\epsilon)\right)$ T-count.

Our sparse state construction is divided into two separate tasks, i.e., (1) the generation of sparse state coefficients via the weighted W-state, and (2) the rearrangement of classical basis states of $W$-state via classically reversible circuits. Through this separation of tasks, we managed to study both tasks independently, and came up with improvements and optimizations of resources that are specific to each task. The tree-based approach of the $W$-state preparation is well suited for many architectures, although the choice of the tree can be optimised on a case-by-case basis e.g.~to reduce the non-Clifford+T gate count. This particular question can be reduced to a \emph{decision problem}, which we conjecture to be \textbf{NP-hard} due to its proximity with combinatorial optimization problems, especially the \emph{partition problem}, which itself is \textbf{NP-complete}. On the other hand, the classically reversible part has been optimized here for non-Clifford gate count and depth, but assuming an all-to-all dependency of the architecture. It may be possible to adapt the algorithm to take specific architectures into account, so that they are as resource-efficient as those of the all-to-all dependency architecture.

%

\section*{Acknowledgments}

The authors acknowledge support from the Pack Quantique project AQEDP, and from the PEPR integrated project EPiQ ANR-22PETQ-0007 part of Plan France 2030.

Renaud Vilmart acknowledges support from the ANR projects TaQC ANR-22-CE47-0012 and HQI ANR-22-PNCQ-0002. 

Sunheang Ty and Chetra Mang acknowledge support by the French government's aid in the framework of PIA (Programme d'Investissement d'Avenir) for Institut de Recherche Technologique SystemX.

\bibliographystyle{plainurl}

\appendix

\section{Details on the algorithm applied to \Cref{ex:permut}}
\label{sec:details-example-permut}

First, for $L$:
\begin{align*}
	\begin{pmatrix}
		1& & & & & & \\
		&1& & & & & \\
		& &1& & & & \\
		& &1&1& & & \\
		&1& & &1& & \\
		1& &1& & &1& \\
		& & & & & &1
	\end{pmatrix}
	\overset{\substack{CX_{0,5}\\CX_{1,4}\\CX_{2,3}}}\longrightarrow
	\begin{pmatrix}
		1& & & & & & \\
		&1& & & & & \\
		& &1& & & & \\
		& & &1& & & \\
		& & & &1& & \\
		& &1& & &1& \\
		& & & & & &1
	\end{pmatrix}
	\overset{\substack{CX_{2,5}}}\longrightarrow
	\begin{pmatrix}
		1& & & & & & \\
		&1& & & & & \\
		& &1& & & & \\
		& & &1& & & \\
		& & & &1& & \\
		& & & & &1& \\
		& & & & & &1
	\end{pmatrix}
\end{align*}
Then, for $U$:
\begin{align*}
	\begin{pmatrix}
		1& & &1&1& &1\\
		&1& &1& &1&1\\
		& &1& &1&1&1\\
		& & & & & & \\
		& & & & & & \\
		& & & & & & \\
		& & & & & & 
	\end{pmatrix}
	\overset{\substack{CCX_{0,1,3}}}\longrightarrow
	\begin{pmatrix}
		1& & &1&1& &1\\
		&1& &1& &1&1\\
		& &1& &1&1&1\\
		& & &1& & &1\\
		& & & & & & \\
		& & & & & & \\
		& & & & & & 
	\end{pmatrix}
	\overset{\substack{CX_{3,0}}}\longrightarrow
	\begin{pmatrix}
		1& & & &1& & \\
		&1& &1& &1&1\\
		& &1& &1&1&1\\
		& & &1& & &1\\
		& & & & & & \\
		& & & & & & \\
		& & & & & & 
	\end{pmatrix}
	\overset{\substack{CCX_{0,2,4}}}\longrightarrow
	\begin{pmatrix}
		1& & & &1& & \\
		&1& &1& &1&1\\
		& &1& &1&1&1\\
		& & &1& & &1\\
		& & & &1& & \\
		& & & & & & \\
		& & & & & & 
	\end{pmatrix}\\
	\overset{\substack{CX_{3,1}\\CX_{4,0}}}\longrightarrow
	\begin{pmatrix}
		1& & & & & & \\
		&1& & & &1& \\
		& &1& &1&1&1\\
		& & &1& & &1\\
		& & & &1& & \\
		& & & & & & \\
		& & & & & & 
	\end{pmatrix}
	\overset{\substack{CCX_{1,2,5}}}\longrightarrow
	\begin{pmatrix}
		1& & & & & & \\
		&1& & & &1& \\
		& &1& &1&1&1\\
		& & &1& & &1\\
		& & & &1& & \\
		& & & & &1& \\
		& & & & & & 
	\end{pmatrix}
	\overset{\substack{CX_{4,2}\\CX_{5,1}}}\longrightarrow
	\begin{pmatrix}
		1& & & & & & \\
		&1& & & & & \\
		& &1& & &1&1\\
		& & &1& & &1\\
		& & & &1& & \\
		& & & & &1& \\
		& & & & & & 
	\end{pmatrix}
	\overset{\substack{CX_{5,2}}}\longrightarrow
	\begin{pmatrix}
		1& & & & & & \\
		&1& & & & & \\
		& &1& & & &1\\
		& & &1& & &1\\
		& & & &1& & \\
		& & & & &1& \\
		& & & & & & 
	\end{pmatrix}\\
	\overset{\substack{CCX_{2,3,6}}}\longrightarrow
	\begin{pmatrix}
		1& & & & & & \\
		&1& & & & & \\
		& &1& & & &1\\
		& & &1& & &1\\
		& & & &1& & \\
		& & & & &1& \\
		& & & & & &1
	\end{pmatrix}
	\overset{\substack{CX_{6,2}}}\longrightarrow
	\begin{pmatrix}
		1& & & & & & \\
		&1& & & & & \\
		& &1& & & & \\
		& & &1& & &1\\
		& & & &1& & \\
		& & & & &1& \\
		& & & & & &1
	\end{pmatrix}
	\overset{\substack{CX_{6,3}}}\longrightarrow
	\begin{pmatrix}
		1& & & & & & \\
		&1& & & & & \\
		& &1& & & & \\
		& & &1& & &\\
		& & & &1& & \\
		& & & & &1& \\
		& & & & & &1
	\end{pmatrix}
\end{align*}

\section{Proofs}

\label{sec:appendix:proofs}

%
%

For the next proof, it will be convenient to work with the unnormalised $W$-state: $\ket{\tilde{W}_n}:=\sqrt n \ket{W_n}$. This version enjoys in particular a nice recursive definition:
\begin{align*}
	\ket{\tilde{W}_1} := \ket1 \qquad\qquad
	\ket{\tilde{W}_n} := \ket{0^k}\ket{\tilde{W}_{n-k}} + \ket{\tilde{W}_k}\ket{0^{n-k}}
\end{align*}
for any $1\leq k < n$.
\begin{align*}
	\ket{W(\rho e^{i\phi})} := e^{i\phi}\ket1 \qquad\qquad
	\ket{W(\vec\alpha_1\cdot\vec\alpha_2)} := \ket{0^k}\ket{\tilde{W}_{n-k}} + \ket{\tilde{W}_k}\ket{0^{n-k}}
\end{align*}

\begin{proof}[Proof of \Cref{prop:semantics-W-prep-weight}]	
	First, let's show that $C(T)$ maps $\ket0$ to $\ket{0^{\otimes n}}$. 
	The base case can easily be checked and comes from the fact that $P(\phi)\ket0 = \ket0$. Then, assuming the property is proven for trees $T_1$ and $T_2$, and using the fact that $F_p\ket0=\ket{00}$:
	\begin{align*}
		
\begin{tikzpicture}
	\begin{pgfonlayer}{nodelayer}
		\node [style=none] (0) at (-3.5, -0.125) {};
		\node [style=none] (1) at (-2, -0.125) {};
		\node [style=none] (2) at (-3.5, -0.875) {};
		\node [style=none] (3) at (-2, -0.875) {};
		\node [style=none] (4) at (-2.75, -0.5) {$C(T_2)$};
		\node [style=none] (5) at (-3.5, -0.5) {};
		\node [style=none] (6) at (-3.75, -0.5) {};
		\node [style=none] (7) at (-1.75, -0.75) {};
		\node [style=none] (8) at (-2, -0.75) {};
		\node [style=none] (9) at (-1.75, -0.25) {};
		\node [style=none] (10) at (-2, -0.25) {};
		\node [style=none] (11) at (-3.5, 0.875) {};
		\node [style=none] (12) at (-2, 0.875) {};
		\node [style=none] (13) at (-3.5, 0.125) {};
		\node [style=none] (14) at (-2, 0.125) {};
		\node [style=none] (15) at (-2.75, 0.5) {$C(T_1)$};
		\node [style=none] (16) at (-3.5, 0.5) {};
		\node [style=none] (17) at (-3.75, 0.5) {};
		\node [style=none] (18) at (-1.75, 0.25) {};
		\node [style=none] (19) at (-2, 0.25) {};
		\node [style=none] (20) at (-1.75, 0.75) {};
		\node [style=none] (21) at (-2, 0.75) {};
		\node [style=none] (22) at (-1.85, 0.6) {$\vdots$};
		\node [style=none] (23) at (-1.85, -0.4) {$\vdots$};
		\node [style=none] (24) at (-5, 0.625) {};
		\node [style=none] (25) at (-3.75, 0.625) {};
		\node [style=none] (26) at (-5, -0.625) {};
		\node [style=none] (27) at (-3.75, -0.625) {};
		\node [style=none] (28) at (-4.375, 0) {$F_{\frac{\ell(T_2)}{\ell(T)}}$};
		\node [style=none] (29) at (-5, 0) {};
		\node [style=none] (30) at (-5.25, 0) {};
		\node [style=none] (31) at (-5.5, 0) {$\ket0$};
		\node [style=none] (32) at (-1, 0) {$\equiv$};
		\node [style=none] (33) at (0.25, -0.125) {};
		\node [style=none] (34) at (1.75, -0.125) {};
		\node [style=none] (35) at (0.25, -0.875) {};
		\node [style=none] (36) at (1.75, -0.875) {};
		\node [style=none] (37) at (1, -0.5) {$C(T_2)$};
		\node [style=none] (38) at (0.25, -0.5) {};
		\node [style=none] (39) at (0, -0.5) {};
		\node [style=none] (40) at (2, -0.75) {};
		\node [style=none] (41) at (1.75, -0.75) {};
		\node [style=none] (42) at (2, -0.25) {};
		\node [style=none] (43) at (1.75, -0.25) {};
		\node [style=none] (44) at (0.25, 0.875) {};
		\node [style=none] (45) at (1.75, 0.875) {};
		\node [style=none] (46) at (0.25, 0.125) {};
		\node [style=none] (47) at (1.75, 0.125) {};
		\node [style=none] (48) at (1, 0.5) {$C(T_1)$};
		\node [style=none] (49) at (0.25, 0.5) {};
		\node [style=none] (50) at (0, 0.5) {};
		\node [style=none] (51) at (2, 0.25) {};
		\node [style=none] (52) at (1.75, 0.25) {};
		\node [style=none] (53) at (2, 0.75) {};
		\node [style=none] (54) at (1.75, 0.75) {};
		\node [style=none] (55) at (1.9, 0.6) {$\vdots$};
		\node [style=none] (56) at (1.9, -0.4) {$\vdots$};
		\node [style=none] (64) at (-0.25, -0.5) {$\ket0$};
		\node [style=none] (65) at (-0.25, 0.5) {$\ket0$};
		\node [style=none] (66) at (2.75, 0) {$\equiv$};
		\node [style=none] (74) at (4, -0.75) {};
		\node [style=none] (75) at (3.75, -0.75) {};
		\node [style=none] (76) at (4, -0.25) {};
		\node [style=none] (77) at (3.75, -0.25) {};
		\node [style=none] (85) at (4, 0.25) {};
		\node [style=none] (86) at (3.75, 0.25) {};
		\node [style=none] (87) at (4, 0.75) {};
		\node [style=none] (88) at (3.75, 0.75) {};
		\node [style=none] (89) at (3.9, 0.6) {$\vdots$};
		\node [style=none] (90) at (3.9, -0.4) {$\vdots$};
		\node [style=none] (91) at (3.5, 0.25) {$\ket0$};
		\node [style=none] (92) at (3.5, 0.75) {$\ket0$};
		\node [style=none] (93) at (3.5, -0.75) {$\ket0$};
		\node [style=none] (94) at (3.5, -0.25) {$\ket0$};
	\end{pgfonlayer}
	\begin{pgfonlayer}{edgelayer}
		\draw (3.center) to (1.center);
		\draw (1.center) to (0.center);
		\draw (0.center) to (2.center);
		\draw (2.center) to (3.center);
		\draw (6.center) to (5.center);
		\draw (8.center) to (7.center);
		\draw (10.center) to (9.center);
		\draw (14.center) to (12.center);
		\draw (12.center) to (11.center);
		\draw (11.center) to (13.center);
		\draw (13.center) to (14.center);
		\draw (17.center) to (16.center);
		\draw (19.center) to (18.center);
		\draw (21.center) to (20.center);
		\draw (27.center) to (25.center);
		\draw (25.center) to (24.center);
		\draw (24.center) to (26.center);
		\draw (26.center) to (27.center);
		\draw (30.center) to (29.center);
		\draw (36.center) to (34.center);
		\draw (34.center) to (33.center);
		\draw (33.center) to (35.center);
		\draw (35.center) to (36.center);
		\draw (39.center) to (38.center);
		\draw (41.center) to (40.center);
		\draw (43.center) to (42.center);
		\draw (47.center) to (45.center);
		\draw (45.center) to (44.center);
		\draw (44.center) to (46.center);
		\draw (46.center) to (47.center);
		\draw (50.center) to (49.center);
		\draw (52.center) to (51.center);
		\draw (54.center) to (53.center);
		\draw (75.center) to (74.center);
		\draw (77.center) to (76.center);
		\draw (86.center) to (85.center);
		\draw (88.center) to (87.center);
	\end{pgfonlayer}
\end{tikzpicture}

	\end{align*}
	We can then show inductively the result. The base case is easily checked as $P(\phi)\ket1 = e^{i\phi}\ket1$. Then, suppose $T=(T_1,T_2)$ with corresponding parameters $\vec\alpha = \vec\alpha_1\cdot\vec\alpha_2$. Using the fact that $F_p\ket1 = \sqrt p\ket{01}+\sqrt{1-p}\ket{10}$, with the induction hypothesis, we get
	\begin{align*}
		
\begin{tikzpicture}
	\begin{pgfonlayer}{nodelayer}
		\node [style=none] (0) at (-4, -0.125) {};
		\node [style=none] (1) at (-2.5, -0.125) {};
		\node [style=none] (2) at (-4, -0.875) {};
		\node [style=none] (3) at (-2.5, -0.875) {};
		\node [style=none] (4) at (-3.25, -0.5) {$C(T_2)$};
		\node [style=none] (5) at (-4, -0.5) {};
		\node [style=none] (6) at (-4.25, -0.5) {};
		\node [style=none] (7) at (-2.25, -0.75) {};
		\node [style=none] (8) at (-2.5, -0.75) {};
		\node [style=none] (9) at (-2.25, -0.25) {};
		\node [style=none] (10) at (-2.5, -0.25) {};
		\node [style=none] (11) at (-4, 0.875) {};
		\node [style=none] (12) at (-2.5, 0.875) {};
		\node [style=none] (13) at (-4, 0.125) {};
		\node [style=none] (14) at (-2.5, 0.125) {};
		\node [style=none] (15) at (-3.25, 0.5) {$C(T_1)$};
		\node [style=none] (16) at (-4, 0.5) {};
		\node [style=none] (17) at (-4.25, 0.5) {};
		\node [style=none] (18) at (-2.25, 0.25) {};
		\node [style=none] (19) at (-2.5, 0.25) {};
		\node [style=none] (20) at (-2.25, 0.75) {};
		\node [style=none] (21) at (-2.5, 0.75) {};
		\node [style=none] (22) at (-2.35, 0.6) {$\vdots$};
		\node [style=none] (23) at (-2.35, -0.4) {$\vdots$};
		\node [style=none] (24) at (-5.5, 0.625) {};
		\node [style=none] (25) at (-4.25, 0.625) {};
		\node [style=none] (26) at (-5.5, -0.625) {};
		\node [style=none] (27) at (-4.25, -0.625) {};
		\node [style=none] (28) at (-4.875, 0) {$F_{\frac{w(T_2)}{w(T)}}$};
		\node [style=none] (29) at (-5.5, 0) {};
		\node [style=none] (30) at (-5.75, 0) {};
		\node [style=none] (31) at (-6, 0) {$\ket1$};
		\node [style=none] (32) at (-1.5, 0) {$\equiv$};
		\node [style=none] (33) at (1, -0.125) {};
		\node [style=none] (34) at (2.5, -0.125) {};
		\node [style=none] (35) at (1, -0.875) {};
		\node [style=none] (36) at (2.5, -0.875) {};
		\node [style=none] (37) at (1.75, -0.5) {$C(T_2)$};
		\node [style=none] (38) at (1, -0.5) {};
		\node [style=none] (39) at (0.75, -0.5) {};
		\node [style=none] (40) at (2.75, -0.75) {};
		\node [style=none] (41) at (2.5, -0.75) {};
		\node [style=none] (42) at (2.75, -0.25) {};
		\node [style=none] (43) at (2.5, -0.25) {};
		\node [style=none] (44) at (1, 0.875) {};
		\node [style=none] (45) at (2.5, 0.875) {};
		\node [style=none] (46) at (1, 0.125) {};
		\node [style=none] (47) at (2.5, 0.125) {};
		\node [style=none] (48) at (1.75, 0.5) {$C(T_1)$};
		\node [style=none] (49) at (1, 0.5) {};
		\node [style=none] (50) at (0.75, 0.5) {};
		\node [style=none] (51) at (2.75, 0.25) {};
		\node [style=none] (52) at (2.5, 0.25) {};
		\node [style=none] (53) at (2.75, 0.75) {};
		\node [style=none] (54) at (2.5, 0.75) {};
		\node [style=none] (55) at (2.65, 0.6) {$\vdots$};
		\node [style=none] (56) at (2.65, -0.4) {$\vdots$};
		\node [style=none] (57) at (0.5, -0.5) {$\ket1$};
		\node [style=none] (58) at (0.5, 0.5) {$\ket0$};
		\node [style=none] (59) at (5.75, -0.125) {};
		\node [style=none] (60) at (7.25, -0.125) {};
		\node [style=none] (61) at (5.75, -0.875) {};
		\node [style=none] (62) at (7.25, -0.875) {};
		\node [style=none] (63) at (6.5, -0.5) {$C(T_2)$};
		\node [style=none] (64) at (5.75, -0.5) {};
		\node [style=none] (65) at (5.5, -0.5) {};
		\node [style=none] (66) at (7.5, -0.75) {};
		\node [style=none] (67) at (7.25, -0.75) {};
		\node [style=none] (68) at (7.5, -0.25) {};
		\node [style=none] (69) at (7.25, -0.25) {};
		\node [style=none] (70) at (5.75, 0.875) {};
		\node [style=none] (71) at (7.25, 0.875) {};
		\node [style=none] (72) at (5.75, 0.125) {};
		\node [style=none] (73) at (7.25, 0.125) {};
		\node [style=none] (74) at (6.5, 0.5) {$C(T_1)$};
		\node [style=none] (75) at (5.75, 0.5) {};
		\node [style=none] (76) at (5.5, 0.5) {};
		\node [style=none] (77) at (7.5, 0.25) {};
		\node [style=none] (78) at (7.25, 0.25) {};
		\node [style=none] (79) at (7.5, 0.75) {};
		\node [style=none] (80) at (7.25, 0.75) {};
		\node [style=none] (81) at (7.4, 0.6) {$\vdots$};
		\node [style=none] (82) at (7.4, -0.4) {$\vdots$};
		\node [style=none] (83) at (5.25, -0.5) {$\ket0$};
		\node [style=none] (84) at (5.25, 0.5) {$\ket1$};
		\node [style=none] (85) at (-0.5, 0) {$\sqrt{\frac{w(T_2)}{w(T)}}$};
		\node [style=none] (86) at (4, 0) {$+\sqrt{\frac{w(T_1)}{w(T)}}$};
		\node [style=none] (87) at (-1.5, -2.5) {$\equiv$};
		\node [style=none] (89) at (2.5, -2.625) {};
		\node [style=none] (91) at (2.5, -3.375) {};
		\node [style=none] (95) at (2.75, -3.25) {};
		\node [style=none] (96) at (2.5, -3.25) {};
		\node [style=none] (97) at (2.75, -2.75) {};
		\node [style=none] (98) at (2.5, -2.75) {};
		\node [style=none] (106) at (2.75, -2.25) {};
		\node [style=none] (107) at (2.5, -2.25) {};
		\node [style=none] (108) at (2.75, -1.75) {};
		\node [style=none] (109) at (2.5, -1.75) {};
		\node [style=none] (110) at (2.65, -1.9) {$\vdots$};
		\node [style=none] (111) at (2.65, -2.9) {$\vdots$};
		\node [style=none] (113) at (2.25, -1.75) {$\ket0$};
		\node [style=none] (142) at (2.25, -2.25) {$\ket0$};
		\node [style=none] (144) at (1.875, -3) {$ W(\vec\alpha_2)$};
		\node [style=none] (145) at (1.25, -2.625) {};
		\node [style=none] (146) at (1.25, -3.375) {};
		\node [style=none] (147) at (7.25, -2.5) {};
		\node [style=none] (148) at (7.25, -1.75) {};
		\node [style=none] (149) at (7.5, -1.875) {};
		\node [style=none] (150) at (7.25, -1.875) {};
		\node [style=none] (151) at (7.5, -2.375) {};
		\node [style=none] (152) at (7.25, -2.375) {};
		\node [style=none] (153) at (7.5, -2.875) {};
		\node [style=none] (154) at (7.25, -2.875) {};
		\node [style=none] (155) at (7.5, -3.375) {};
		\node [style=none] (156) at (7.25, -3.375) {};
		\node [style=none] (157) at (7.4, -3.025) {$\vdots$};
		\node [style=none] (158) at (7.4, -2.025) {$\vdots$};
		\node [style=none] (159) at (7, -3.375) {$\ket0$};
		\node [style=none] (160) at (7, -2.875) {$\ket0$};
		\node [style=none] (162) at (6, -2.5) {};
		\node [style=none] (163) at (6, -1.75) {};
		\node [style=none] (164) at (6.625, -2.125) {$ W(\vec\alpha_1)$};
		\node [style=none] (165) at (-0.5, -2.5) {$\sqrt{\frac{w(T_2)}{w(T)}}$};
		\node [style=none] (166) at (4, -2.5) {$+\sqrt{\frac{w(T_1)}{w(T)}}$};
	\end{pgfonlayer}
	\begin{pgfonlayer}{edgelayer}
		\draw (3.center) to (1.center);
		\draw (1.center) to (0.center);
		\draw (0.center) to (2.center);
		\draw (2.center) to (3.center);
		\draw (6.center) to (5.center);
		\draw (8.center) to (7.center);
		\draw (10.center) to (9.center);
		\draw (14.center) to (12.center);
		\draw (12.center) to (11.center);
		\draw (11.center) to (13.center);
		\draw (13.center) to (14.center);
		\draw (17.center) to (16.center);
		\draw (19.center) to (18.center);
		\draw (21.center) to (20.center);
		\draw (27.center) to (25.center);
		\draw (25.center) to (24.center);
		\draw (24.center) to (26.center);
		\draw (26.center) to (27.center);
		\draw (30.center) to (29.center);
		\draw (36.center) to (34.center);
		\draw (34.center) to (33.center);
		\draw (33.center) to (35.center);
		\draw (35.center) to (36.center);
		\draw (39.center) to (38.center);
		\draw (41.center) to (40.center);
		\draw (43.center) to (42.center);
		\draw (47.center) to (45.center);
		\draw (45.center) to (44.center);
		\draw (44.center) to (46.center);
		\draw (46.center) to (47.center);
		\draw (50.center) to (49.center);
		\draw (52.center) to (51.center);
		\draw (54.center) to (53.center);
		\draw (62.center) to (60.center);
		\draw (60.center) to (59.center);
		\draw (59.center) to (61.center);
		\draw (61.center) to (62.center);
		\draw (65.center) to (64.center);
		\draw (67.center) to (66.center);
		\draw (69.center) to (68.center);
		\draw (73.center) to (71.center);
		\draw (71.center) to (70.center);
		\draw (70.center) to (72.center);
		\draw (72.center) to (73.center);
		\draw (76.center) to (75.center);
		\draw (78.center) to (77.center);
		\draw (80.center) to (79.center);
		\draw (91.center) to (89.center);
		\draw (96.center) to (95.center);
		\draw (98.center) to (97.center);
		\draw (107.center) to (106.center);
		\draw (109.center) to (108.center);
		\draw (146.center) to (145.center);
		\draw (146.center) to (91.center);
		\draw (145.center) to (89.center);
		\draw (148.center) to (147.center);
		\draw (150.center) to (149.center);
		\draw (152.center) to (151.center);
		\draw (154.center) to (153.center);
		\draw (156.center) to (155.center);
		\draw (163.center) to (162.center);
		\draw (163.center) to (148.center);
		\draw (162.center) to (147.center);
	\end{pgfonlayer}
\end{tikzpicture}

	\end{align*}
	Noticing that $w(T_i)=w(\alpha_i)$, $w(T)=w(T_1)+w(T_2)$ and $\ell(T_i) = |\alpha_i|$, the circuit hence maps $\ket1$ to:
	\[\frac1{\sqrt{w(T)}}\left(\sqrt{w(\vec\alpha_2)}\ket{0^{\ell(T_1)}} \ket{W(\vec\alpha_2)} + \sqrt{w(\vec\alpha_1)} \ket{W(\vec\alpha_1)}\ket{0^{\ell(T_2)}}\right) = \ket{W(\vec\alpha)}\]
\end{proof}

\begin{proof}[Proof of \Cref{thm:Wn-weight}]
Let $t$ be a complete tree whose leaves are the squares of the amplitudes of $\vec\alpha$. Then $C(t)\ket1$ is a circuit that builds $\ket{W(\vec\alpha)}$ by \Cref{prop:semantics-W-prep-weight}, and whose height is $\mathcal O(\log_2(n))$ since the tree is complete.
\end{proof}

The proof of $\Cref{prop:Fp-resources}$ requires a trivial but useful result on the square of odd numbers:

\begin{lemma}
\label{lem:odd-square-mod-4}
Suppose $a\in\mathbb Z$ is odd. Then $a^2 = 1 \bmod 4$.
\end{lemma}

\begin{proof}
Write $a = 2\alpha+1$. Then $a^2 = 4\alpha^2+4\alpha+1 = 1\bmod 4$.
\end{proof}

We can then show:

\begin{proof}[Proof of \Cref{prop:Fp-resources}]
When $p\in\{0,\frac12,1\}$, we already know how to build a Clifford+T circuit that implements $F_p$. Let's now suppose that $F_p$ is in Clifford+T, and show that it forces $p\in\{0,\frac12,1\}$.

Result from \cite{Giles2013Exact} shows that any Clifford+T circuit implements a matrix with coefficients in $\mathbb Z[\frac1{\sqrt2},i]$. We hence have $\sqrt p\in\mathbb Z[\frac1{\sqrt2}]$ and  $\sqrt{1-p}\in\mathbb Z[\frac1{\sqrt2}]$. Suppose:
\[\sqrt p = \frac 1{2^k}(a+\sqrt2b)\qquad\qquad\sqrt{1-p} = \frac 1{2^\ell}(c+\sqrt2d)\]
for $k,\ell\in\mathbb N$ and $a,b,c,d\in\mathbb Z$. Unless both $a$ and $b$ are null (which is the trivial case $p=0$), we assume that at leat one of $a$ and $b$ is odd, since otherwise we could simplify the decomposition by dividing by $2$. Similarly, at least one of $c$ and $d$ is odd (unless $p=1$). Without loss of generality, we can assume that $k\geq l$.

We have:
\begin{align*}
	4^k &= 4^k\left(\sqrt p^2 + \sqrt{1-p}^2\right) = a^2+2b^2+2\sqrt2ab + 4^{k-\ell}(c^2+2d^2+2\sqrt2cd)\\
	&\implies\begin{cases}
		a^2+2b^2 + 4^{k-\ell}(c^2+2d^2) = 4^{k}\\
		ab+4^{k-\ell}cd = 0
	\end{cases}
\end{align*}
We can already show that $\ell = k$. Indeed, assuming $k>\ell$, the first equation gives:
\[a^2+2b^2 = 0 \bmod4\]
which is a contradiction when either $a$ or $b$ (or both) is odd, using \Cref{lem:odd-square-mod-4}. We hence have:
\[\begin{cases}
	a^2+2b^2 + c^2+2d^2 = 4^k\\
	ab+cd = 0
\end{cases}\]
We can show that the case $k=0$ has no solution unless $p\in\{0,1\}$. Assuming $k=0$, the first equation yields $a^2+c^2=1\bmod2$, i.e.~one of $a$ or $c$ is odd, the other is even. Suppose w.l.o.g.~that $a$ is odd. Since $c$ is even, $d$ has to be odd (or $\sqrt{1-p}=0$). The first equation then yields $1+2b^2+0+2 = 1\bmod 4$, hence $b$ is odd. The fact that both $a$ and $b$ are odd while $c$ is even contradicts the second equation $ab+cd=0$. We can now assume $k\geq1$.

We then distinguish based on the parity of $a$.
\begin{itemize}
	\item Suppose $a$ is odd. Then the first equation $\bmod 2$ gives that $c$ is odd as well. Then, the first equation $\bmod 4$ together with \Cref{lem:odd-square-mod-4} gives $2(b^2+d^2)=2\bmod4$, i.e.~$b^2+d^2=1\bmod2$. Hence exactly one of $b$ and $d$ is odd, the other is even. Since $a$ and $c$ are odd, this contradicts the second equation $ab+cd=0$.
	\item Suppose $a$ is even. Then $b$ is odd. The first equation $\bmod 2$ shows that $c$ is even, which implies $d$ is odd. Let's then write:
	\[a = 2\alpha,\qquad b=2\beta+1, \qquad c=2\gamma, \qquad d = 2\delta+1\]
	The first equation then yields:
	\begin{align*}
		4\alpha^2 + 8\beta^2+8\beta+2 + 4\gamma^2 + 8\delta^2+8\delta+2 = 4^k\\
		\implies \alpha^2+\gamma^2+2(\beta^2+\beta+\gamma^2+\gamma)+1 = 4^{k-1}
	\end{align*}
	The case $k=1$ has to be handled differently:
	\[k=1\implies \alpha^2+\gamma^2+2(\beta^2+\beta+\gamma^2+\gamma)=0\]
	This can only happen when $\beta^2+\beta+\gamma^2+\gamma$ is non-positive, which can only happen when $\beta=\delta=0$. Then $\alpha=\gamma=0$. This case corresponds to the case $p=\frac12$.
	
	When $k\geq2$, we get $\alpha^2+\gamma^2 = 1 \bmod 2$ from the first equation. From the second equation, we get:
	\begin{align*}
		4\alpha\beta+2\alpha + 4\gamma\delta+2\gamma = 0 \implies 2\alpha\beta +\alpha + 2\gamma\delta+\gamma = 0 \implies  \alpha+\gamma = 0 \bmod 2
	\end{align*}
	which is a contradiction.
\end{itemize}
Hence, the only cases where $F_p$ can be built in Clifford+T are exactly when $p\in\{0,\frac12,1\}$.
\end{proof}

\begin{proof}[Proof of \Cref{thm:Wn}]
	Again by taking a complete tree, one can notice that at most one subtree per level is not itself complete, and that when all leaves have weight $1$, a complete tree $t=(t_1,t_2)$ is such that $w(t_2)/w(t_1) = 1/2$ (so $F_p$ is in Clifford+T). Hence we have at most $\mathcal O(\log_2(n))$ non-Clifford+T gates.
\end{proof}

\begin{proof}[Proof of \cref{thm:Wn-weight-from-compact}]
	Let $t := \lceil\log_2(n)\rceil$ and $\ket{\vec\alpha}$ be the $t$-qubit state $\ket{\vec\alpha} := \sum_i \alpha_i \ket i$. We build the following circuit:
	\[
\begin{tikzpicture}
	\begin{pgfonlayer}{nodelayer}
		\node [style=none] (0) at (-1, 1.375) {};
		\node [style=none] (1) at (-1, 0.125) {};
		\node [style=none] (2) at (-0.5, 1.375) {};
		\node [style=none] (3) at (-0.5, 0.125) {};
		\node [style=none] (4) at (-0.75, 0.75) {$\ket{\vec\alpha}$};
		\node [style=none] (5) at (-0.5, 0.25) {};
		\node [style=none] (6) at (0, 0.25) {};
		\node [style=none] (7) at (-0.5, 1.25) {};
		\node [style=none] (8) at (0, 1.25) {};
		\node [style=none, rotate=90] (10) at (-0.275, 0.75) {...};
		\node [style=none] (11) at (1.5, 0.75) {$t$};
		\node [style=none] (12) at (-0.5, -0.25) {};
		\node [style=none] (13) at (0, -0.25) {};
		\node [style=none] (14) at (0, 1.375) {};
		\node [style=none] (15) at (1, 1.375) {};
		\node [style=none] (16) at (0, -1.375) {};
		\node [style=none] (17) at (1, -1.375) {};
		\node [style=none] (18) at (0.5, 0) {$U_f$};
		\node [style=none] (19) at (-0.75, -0.25) {$\ket0$};
		\node [style=none] (20) at (-0.5, -1.25) {};
		\node [style=none] (21) at (0, -1.25) {};
		\node [style=none] (22) at (-0.75, -1.25) {$\ket0$};
		\node [style=none] (23) at (1, 0.25) {};
		\node [style=none] (24) at (1.5, 0.25) {};
		\node [style=none] (25) at (1, 1.25) {};
		\node [style=none] (26) at (1.5, 1.25) {};
		\node [style=none, rotate=90] (27) at (1.225, 0.75) {...};
		\node [style=none] (28) at (1, -1.25) {};
		\node [style=none] (29) at (1.5, -1.25) {};
		\node [style=none] (30) at (1, -0.25) {};
		\node [style=none] (31) at (1.5, -0.25) {};
		\node [style=none, rotate=90] (32) at (1.225, -0.75) {...};
		\node [style=none, rotate=90] (33) at (-0.5, -0.75) {...};
		\node [style=none] (34) at (1.75, -0.75) {$n-t$};
	\end{pgfonlayer}
	\begin{pgfonlayer}{edgelayer}
		\draw (0.center) to (1.center);
		\draw (3.center) to (1.center);
		\draw (0.center) to (2.center);
		\draw (2.center) to (3.center);
		\draw (7.center) to (8.center);
		\draw (6.center) to (5.center);
		\draw (14.center) to (15.center);
		\draw (16.center) to (17.center);
		\draw (16.center) to (14.center);
		\draw (17.center) to (15.center);
		\draw (12.center) to (13.center);
		\draw (20.center) to (21.center);
		\draw (25.center) to (26.center);
		\draw (24.center) to (23.center);
		\draw (30.center) to (31.center);
		\draw (29.center) to (28.center);
	\end{pgfonlayer}
\end{tikzpicture}
\]
	In this circuit, $\vec\alpha$ is synthesised using \cref{thm:Gosset}, resulting in $\mathcal O\left(\sqrt{2^n\log_2(1/\epsilon)}+\log_2(1/\epsilon)\right)$ T-gates, ancillas and depth. $U_f$ on the other hand is the dagger of the circuit built using \cref{thm:permut} with $\chi_i:=i$ i.e.~integer $i$ written on $t$ bits, resulting in an exact, ancilla-free circuit with $\mathcal O(n)$ T-count and depth.
	
	The overall $T$-count and depth are hence $\mathcal O(n+\log_2(1/\epsilon))$, while the initialised qubits may be used as ancillas for the preparation of $\ket{\vec\alpha}$, resulting in $\mathcal O\left(\sqrt{n\log_2(1/\epsilon)}+\log_2(1/\epsilon)\right) - n + \log_2(n)$ ancillas.
\end{proof}

\begin{proof}[Proof of \Cref{prop:lower-bound}]
The number of $T$-gates required to build an $n$-qubit state $\ket\psi$ is bounded below by its ``stabiliser nullity'' \cite{Beverland2020Lower}, defined as $\nu(\ket\psi) = n - \log_2(|\operatorname{Stab}(\ket\psi)|)$ where $\operatorname{Stab}(\ket\psi)$ is the set of Pauli strings that stabilise $\ket\psi$. It can be shown that $\operatorname{Stab}(\ket{W_s}) = \{I^{\otimes s}, -Z^{\otimes s}\}$ and $\operatorname{Stab}(\ket{0^{\otimes k}}) = \{Z^v\}_{v\in\{0,1\}^k}$. We then distinguish between the cases where $s\leq n$ and $s>n$.

One special case of a sparse state with $s\leq n$ is $\ket{W_s}\otimes\ket{0^{\otimes n-s}}$, which has $2^{n-s+1}$ stabilisers. Its stabiliser nullity is hence $s-1=\mathbf{\Omega}(\min(s,n))$.

For the case of sparse states with $s > n$, it suffices to find an $s$-sparse state $\ket\psi$ with no stabiliser other than the identity. Let's consider the state:
\[\ket\psi := \sum_{k=0}^n \sum_{\substack{x\in\{0,1\}^n\\|x|=k}} \alpha_x\ket{x}\]
where $\alpha_x \geq0$, $|x|\leq1 \implies \alpha_x\neq 0$, and $x\neq y \implies \alpha_x\neq\alpha_y$. One can easily check that $\operatorname{Stab}(\ket\psi) = \{I^{\otimes n}\}$. As a result, its stabiliser nullity is $n = \mathbf{\Omega}(\min(s,n))$.
\end{proof}

\section{Synthesis Algorithms}

\label{sec:appendix:algos}

\begin{algorithm2e}[!ht]
	\caption{AntiDiagRemoval}
	\label{alg:anti-diag-rem}
	\KwData{$i$ an index, and $U$ an upper triangular matrix, with $n$ rows and $m$ columns, such that the $j$th anti-diagonal is null for $j<i$.}
	\KwResult{A set of (disjoint) pairs of indices, representing a round of parallel row additions $R$ that remove anti-diagonal $i$ when applied to $U$, if successful; empty if unsuccessful. Updated $U$ if successful. A boolean stating if the diagonal removal was successful.}
	$R\gets \{\}$\;
	$k\gets 0$\;
	\If{$i\geq m$}{$k\gets i-m-1$}
	\While{$k<i-k$}{
		\If{$U[k,i-k]=1$}{
			\eIf{$U[i-k,i-k]=1$}{$R\gets R\cup\{(i-k,k)\}$}{\Return{$\{\}$, $U$, \textsf{ff}}}
		}
		$k\gets k+1$
	}
	$U\gets$ updated $U$ with row additions from $R$\;
	\Return{$R$, $U$, \textsf{tt}}
\end{algorithm2e}

\begin{algorithm2e}[!ht]
	\caption{UpElimComp}
	\label{alg:upper-elim-comp}
	\KwData{$U$ an upper triangular matrix, with $m$ columns, and with pivots on the left.}
	\KwResult{A list of sets of (disjoint) tuples of indices, representing a sequence of parallel row multiplications (if tuple is triplet) and row additions (if tuple is pair). A list of pairs, representing a sequence of swaps.}
	$R\gets []$  \Comment*[r]{resulting list of sets of tuples}
	$P \gets []$ \Comment*[r]{resulting permutation on columns}
	$d\gets 1$   \Comment*[r]{index of first anti-diagonal not set to $0$}
	$t\gets 0$   \Comment*[r]{index of first row of weight $\geq2$}
	$i_0\gets$ first index s.t.~$U[i_0,i_0]=0$\;
	\For{$i_0\leq i<m$}{
		$b\gets \textsf{tt}$\;
		\While{$b$ and $d<t+i$}{
			$r, U, b \gets \operatorname{AntiDiagRemoval}(d,U)$\;
			\If{$b$}{
				$R\gets R@[r]$\;
				$d\gets d+1$
			}
			$t\gets$ first index s.t.~row $t$ of $U$ has weight $\geq 2$
		}
		$j\gets$ first index $\geq i$ s.t.~$U_{t,j}=1$\;
		\If{$i\neq j$}{
			swap columns $i$ and $j$\;
			$P\gets P@[(i,j)]$
		}
		$t'\gets$ second index s.t.~$U_{t',i}=1$\;
		$U[i] \gets U[i] \oplus U[t]\odot U[t']$ \Comment*[r]{with $\odot$ the element-wise product}
		$R\gets R@[(t,t',i)]$
	}
	\While{$d<2m$}{
		$r, U, \underline{~~} \gets \operatorname{AntiDiagRemoval}(d,U)$\;
		$R\gets R@[r]$\;
		$d\gets d+1$
	}
	\Return{$R$, $P$}
\end{algorithm2e}

\end{document}